  \def \VersionLong {} 
  \def \AuthorVersion {}

\PassOptionsToPackage{svgnames,table}{xcolor}

\documentclass[conference]{IEEEtran}
\IEEEoverridecommandlockouts

 \usepackage{algorithmic}
\usepackage{graphicx}

\usepackage[utf8]{inputenc}

\usepackage[ruled,vlined,linesnumbered]{algorithm2e}

\usepackage{subcaption}

\usepackage{paralist} 

\newenvironment{ienumeration}
	{\ifdefined\VersionLong\begin{enumerate}\else\begin{inparaenum}[\itshape i\upshape)]\fi}
	{\ifdefined\VersionLong\end{enumerate}\else\end{inparaenum}\fi}

\usepackage{amsmath} 
\usepackage{amssymb} 

\ifdefined\AuthorVersion
	\usepackage[backend=biber,backref=true,style=alphabetic,url=false,doi=true,defernumbers=true,sorting=anyt,maxnames=99]{biblatex} 
	\renewbibmacro*{doi+eprint+url}{%
		\iftoggle{bbx:doi}
			{\color{black!40}\footnotesize\printfield{doi}}
			{}%
		\newunit\newblock
		\iftoggle{bbx:eprint}
			{\usebibmacro{eprint}}
			{}%
		\newunit\newblock
		\iftoggle{bbx:url}
			{\usebibmacro{url+urldate}}
			{}%
	}

\fi

\usepackage{tikz}
\usetikzlibrary{arrows,positioning,automata,shadows}

\tikzstyle{every node}=[initial text=]
\tikzstyle{location}=[rectangle, rounded corners, minimum size=12pt, draw=black, fill=blue!10, inner sep=2pt]
\tikzstyle{invariant}=[yshift=3, rectangle, draw=black, fill=white, text=black, inner sep=1pt]
\tikzstyle{final}=[double]
\tikzstyle{nongardee}=[dashed]
\tikzstyle{success}=[fill=green!50]
\tikzstyle{failure}=[fill=red!50]
\tikzstyle{urgent}=[fill=yellow, densely dotted]
\tikzstyle{invariant}=[draw=black, dotted, inner sep=1pt] 
\tikzstyle{gentil}=[fill=yellow]
\tikzstyle{pasgentil}=[fill=white]

\newcommand{\gentilinline}[1]{\ensuremath{\colorbox{yellow}{\ensuremath{#1}}}}

	\definecolor{coloract}{rgb}{0.50, 0.70, 0.30}
	\definecolor{colorclock}{rgb}{0.4, 0.4, 1}
	\definecolor{colordisc}{rgb}{1, 0, 1}
	\definecolor{colorloc}{rgb}{0.4, 0.4, 0.65}
	\definecolor{colorparam}{rgb}{1, 0.6, 0.0}

\newcommand{\styleact}[1]{\ensuremath{\textcolor{coloract}{\mathrm{#1}}}}
\newcommand{\styleclock}[1]{\ensuremath{\textcolor{colorclock}{#1}}}

\newcommand{\styleloc}[1]{\ensuremath{\mathrm{#1}}}
\newcommand{\styleparam}[1]{\ensuremath{\textcolor{colorparam}{#1}}}

\newcommand{\init}{_0}

\newcommand{\A}{\ensuremath{\mathcal{A}}}
\newcommand{\Actions}{\Sigma}
\newcommand{\action}{\ensuremath{a}}
 \newcommand{\AP}{\ensuremath{\mathit{AP}}}
 \newcommand{\atomicprop}{\ensuremath{\textit{ap}}}

\newcommand{\C}{C}

\newcommand{\Clock}{\mathbb{X}} 
\newcommand{\ClockCard}{H} 
\newcommand{\clock}{x} 
\newcommand{\clocky}{y} 
\newcommand{\clockz}{z} 
\newcommand{\clockval}{w} 
\newcommand{\ClocksZero}{\vec{0}}
 \newcommand{\cm}{\ensuremath{\mathtt{c}}} 
 \newcommand{\cms}{\ensuremath{\mathtt{q}}} 
 \newcommand{\cmshalt}{\ensuremath{\cms_\textrm{halt}}} 
 \newcommand{\twoCM}{two-counter machine}

\newcommand{\compOp}{\bowtie}
\newcommand{\compOpLeq}{\triangleleft}
\newcommand{\CTrue}{\text{true}}
 \newcommand{\locerror}{\ensuremath{\loc_{\mathtt{error}}}}
 \newcommand{\lochalt}{\ensuremath{\loc_{\mathtt{halt}}}}

\newcommand{\edge}{e}
\newcommand{\Edges}{E}
\newcommand{\formule}{\ensuremath{\varphi}} 
\newcommand{\longuefleche}[1]{\stackrel{#1}{\longrightarrow}}
\newcommand{\longueflecheRel}[1]{\stackrel{#1}{\mapsto}}

\newcommand{\flecheRel}{{\rightarrow}}

\newcommand{\grandn}{{\mathbb N}}
\newcommand{\grandq}{{\mathbb Q}}
\newcommand{\grandqplus}{\grandq_{+}} 
\newcommand{\grandr}{{\mathbb R}}
\newcommand{\grandrplus}{\grandr_{+}} 
\newcommand{\grandz}{{\mathbb Z}}
\newcommand{\guard}{g}

  \newcommand{\invariant}{I}
\newcommand{\K}{K}

\newcommand{\Label}{\mathit{lb}} 

\newcommand{\loc}{\ell} 
\newcommand{\locinit}{\loc\init}
\newcommand{\Loc}{L} 
\newcommand{\locfinal}{\loc_f}
 \newcommand{\Lab}{\mathbf{L}} 

\newcommand{\Param}{\mathbb{P}} 
\newcommand{\param}{p} 
\newcommand{\parama}{a} 
\newcommand{\pcc}{\mathit{pcc}} 
\newcommand{\ParamCard}{M} 
\newcommand{\pval}{v} 
\newcommand{\pvalone}{\pval_{1}} 

\newcommand{\resets}{R}

\newcommand{\T}[1]{\ensuremath{T(#1)}}

\newcommand{\uPTAs}{\textcolor{colorok}{U-PTAs}}
\newcommand{\PTAi}{\textcolor{colorok}{PTA\ensuremath{_{I}}}}
\newcommand{\PTAis}{\textcolor{colorok}{PTAs\ensuremath{_{I}}}}
\newcommand{\PTAiu}{\textcolor{colorok}{PTA\ensuremath{_I^U}}}
\newcommand{\PTAius}{\textcolor{colorok}{PTAs\ensuremath{_I^U}}}

\newcommand{\sinit}{s\init} 
\renewcommand{\S}{S}
\newcommand{\state}{\ensuremath{s}} 
\newcommand{\States}{S} 
\newcommand{\varrun}{\rho} 

 \newcommand{\gentil}{{\gentilinline{\circ}}} 
 \newcommand{\mechant}{\ensuremath{\circ}}

\newcommand{\reset}[2]{\ensuremath{[#1]_{#2}}}
\newcommand{\valuate}[2]{\ensuremath{#2(#1)}}

\newcommand{\imitator}{\textsf{IMITATOR}}

\newcommand{\Time}{\mathsf{time}}

\newcommand{\defProblem}[3]
{%
\noindent\fcolorbox{black}{blue!15}{
	\begin{minipage}{.95\columnwidth}
		\textbf{#1 problem:}\\
		\textsc{Input}: #2\\
		\textsc{Problem}: #3
	\end{minipage}
}

	\smallskip

}

\ifdefined\VersionWithComments
	\usepackage{draftwatermark}
	\SetWatermarkText{draft}
	\SetWatermarkScale{3}
	\SetWatermarkColor[gray]{0.9}
\fi

\definecolor{darkblue}{rgb}{0.0,0.0,0.6}
\definecolor{darkgreen}{rgb}{0, 0.5, 0}
\definecolor{darkpurple}{rgb}{0.7, 0, 0.7}
\definecolor{violetcurie}{RGB}{115,26,67}
\definecolor{forestgreen}{rgb}{0.13,0.54,0.13}
\definecolor{darkblue}{rgb}{0, 0, 0.7}

\usepackage[
		pdfauthor={Étienne André, Didier Lime and Mathias Ramparison},%
		pdftitle={On the expressive power of invariants in parametric timed automata},
		colorlinks=true,
	\ifdefined \VersionWithComments
		pagebackref=true,
	\fi
		citecolor   = blue!50!blue,
		linkcolor   = darkblue,
		urlcolor    = blue!50!green,
	]{hyperref}

\usepackage[capitalise,english,nameinlink]{cleveref} 
\crefname{line}{\text{line}}{\text{lines}} 

\ifdefined \VersionLong
	\newcommand{\LongVersion}[1]{#1}
	\newcommand{\ShortVersion}[1]{}
\else
	\newcommand{\LongVersion}[1]{}
	\newcommand{\ShortVersion}[1]{#1}
\fi

\newcommand{\VeryLongVersion}[1]{}

\usepackage{amsthm}
\newtheorem{lemma}{Lemma}

\newtheorem{theorem}{Theorem}

\newtheorem{definition}{Definition}
\newtheorem{example}{Example}

\ifdefined\VersionWithComments
	\usepackage[colorinlistoftodos,textsize=footnotesize]{todonotes}
\else
	\usepackage[disable]{todonotes}
\fi
\newcommand{\gennote}[3]{\todo[linecolor=#2,backgroundcolor=#2!25,bordercolor=#2]{#3: #1}\xspace}
\newcommand{\ea}[1]{\gennote{#1}{blue}{ÉA}}
\newcommand{\mr}[1]{\gennote{#1}{orange}{$\mathbb{MR}$}}
\newcommand{\dl}[1]{{\gennote{#1}{purple}{DL}}}
\newcommand{\instructions}[1]{{\gennote{\bfseries #1}{red}{Instructions}}}

\ifdefined \VersionWithComments
	\newcommand{\todoinline}[1]{\mbox{}{\color{red}{\textbf{TODO}\ifx#1\\\else:\ \fi #1}}} 
\else
	\newcommand{\todoinline}[1]{}
\fi

\usepackage{verbatim} 

\newcommand{\styleTCTL}[1]{\ensuremath{\mathsf{#1}}}
\newcommand{\EF}{\styleTCTL{EF}}
\newcommand{\EG}{\styleTCTL{EG}}
\newcommand{\AF}{\styleTCTL{AF}}
\newcommand{\AG}{\styleTCTL{AG}}
\newcommand{\notresuperformule}{\ensuremath{\EG \AF_{=0}}}

\ifdefined \VersionWithComments
 	\definecolor{colorok}{RGB}{80,80,150}
\else
	\definecolor{colorok}{RGB}{0,0,0}
\fi

\newcommand{\eg}{\textcolor{colorok}{e.\,g.,}\xspace}
\newcommand{\ie}{\textcolor{colorok}{i.\,e.,}\xspace}
\newcommand{\st}{\textcolor{colorok}{s.t.}\xspace}

\newcommand{\wlogen}{\textcolor{colorok}{w.l.o.g.}\xspace}
\newcommand{\wrt}{\textcolor{colorok}{w.r.t.}\xspace}

\begin{document}

\title{On the expressive power of invariants in parametric timed automata\thanks{%
	\ifdefined\AuthorVersion
	This is the author version of the manuscript of the same name published in the proceedings of the 24th International Conference on Engineering of Complex Computer Systems (\href{https://www.formal-analysis.com/iceccs/2019/}{ICECCS 2019}).
	The final version is available at 
		\href{https://ieeexplore.ieee.org/}{\nolinkurl{ieeexplore.ieee.org}}.
	\fi%
	This work is partially supported by the ANR national research program PACS (ANR-14-CE28-0002)
	and
	by ERATO HASUO Metamathematics for Systems Design Project (No.\ JPMJER1603), JST.
	}
}

\author{\IEEEauthorblockN{\'Etienne Andr\'e}
\IEEEauthorblockA{\textit{Universit\'e Paris 13, LIPN, CNRS,} \\
\textit{UMR 7030, F-93430,}\\
Villetaneuse, France\\
JFLI, CNRS, Tokyo, Japan\\
National Institute of Informatics, Tokyo, Japan
}
\and
\IEEEauthorblockN{Didier Lime}
\IEEEauthorblockA{\textit{\'Ecole Centrale de Nantes, LS2N, CNRS,} \\
\textit{UMR 6004,}\\
Nantes, France
}
\and
\IEEEauthorblockN{Mathias Ramparison}
\IEEEauthorblockA{\textit{Universit\'e Paris 13, LIPN, CNRS,} \\
\textit{UMR 7030, F-93430}\\
Villetaneuse, France
}
}

\maketitle

\thispagestyle{plain}

\begin{abstract}
The verification of systems combining hard timing constraints with concurrency is challenging.
This challenge becomes even harder when some timing constants are missing or unknown.
Parametric timed formalisms, such as parametric timed automata (PTAs), tackle the synthesis of such timing constants (seen as parameters) for which a property holds.
  Such formalisms are highly expressive, but also\mr{enlevé: highly} undecidable\dl{juste undecidable? Que veut dire le highly?}\ea{je sais pas trop ; on peut enlever ? 'fin l'idée c'était que plein de choses étaient indécidables}, and few decidable subclasses were proposed.
We propose here a syntactic restriction on PTAs consisting in removing guards (constraints on transitions) to keep only invariants (constraints on locations).
While this restriction preserves the expressiveness of PTAs (and therefore their undecidability), an additional restriction on the type of constraints allows to not only prove decidability, but also to perform the exact synthesis of parameter valuations satisfying reachability.
This formalism, that seems trivial at first sight as it benefits from the decidability of the reachability problem with a better complexity than Timed Automata (TAs), suffers from the undecidability of the whole TCTL logic that TAs, on the contrary enjoy.
We believe our formalism allows for an interesting trade-off between decidability and practical expressiveness and is therefore promising.
We show its applicability in a small case study.\ea{OK?}

\end{abstract}

\instructions{ICECCS 10 pages double column, including references and appendix}

\ea{hello}
\dl{hello}
\mr{hello}


 \todo{This is the version with comments. To disable comments, comment out line~3 in the \LaTeX{} source.}

\section{Introduction}\label{section:introduction}

The verification of systems combining hard timing constraints with concurrency is challenging.
This challenge becomes even harder when some timing constants are missing or unknown.
Parametric timed formalisms tackle the synthesis of such timing constants (seen as parameters) for which a property holds.
A well-known such formalism is parametric timed automata~\cite{AHV93}, a formalism extending finite-state automata with clocks~\cite{AD94}, that can be compared to either integer constants or to integer-valued or real-valued parameters along guards (over transitions) or in invariants (in locations).
Such formalisms are highly expressive, but also highly undecidable, and only a few decidable subclasses were proposed.


In the PTA literature, the main problem studied is \EF{}-emptiness (``is the set of valuations for which a given location is reachable for at least one run empty?''): it is ``robustly'' undecidable in the sense that, even when varying the setting, undecidability is preserved.
For example, \EF{}-emptiness is undecidable even for a single bounded parameter~\cite{Miller00}, even for a single rational-valued or integer-valued parameter~\cite{BBLS15}, even with only one clock compared to parameters~\cite{Miller00}, or with strict constraints only~\cite{Doyen07} (see~\cite{Andre19STTT} for a survey).
%
Decidability can be obtained using two main directions.

First, reducing the number of clocks may lead to decidability:
	for example, decidability is ensured in some restrictive settings such as over discrete time with a single parametric clock (\ie{} compared to parameters in at least one guard)~\cite{AHV93}, or over discrete or dense time with one parametric clock and arbitrarily many non-parametric clocks~\cite{BO14,BBLS15}, or over discrete time with two parametric clocks and a single parameter~\cite{BO14}.
	But the practical power of these restrictive settings remains unclear.

Second, restricting the syntax may also lead to decidability, notably on two main subclasses:
in~\cite{HRSV02}, \emph{L/U-PTAs} are proposed as a subclass where parameters are partitioned into upper-bound parameters (only compared to clocks as upper-bounds, \ie{} of the form $\clock > \param$ or $\clock \geq \param$, where $\clock$ is a clock and~$\param$ a parameter\dl{pas clair: de quel côté on met paramètres et horloges? je mettrais directement ``$x<p$ or $x\leq p$ where x is a clock and p a parameter''}\ea{fixed}) and lower-bound parameters.
While L/U-PTAs benefit from the decidability of \EF{}-emptiness~\cite{JLR15,BlT09}, \AF{}-emptiness (``is the set of valuations for which a given location is reachable for all runs empty?'')\ is undecidable~\cite{JLR15}; even more annoying, it is impossible to achieve exact synthesis, even for \EF{}: that is, it is not possible in general to compute the set of parameter valuations for which a given location is reachable.
A second restriction of the syntax is proposed in~\cite{ALR19}: in \emph{reset-PTAs}, whenever a clock is compared to a parameter, all clocks must be reset (possibly to parameters, which extends the original PTA syntax).
While exact synthesis over bounded rational-valued parameters can be achieved for \EF{}, resetting all clocks as soon as one clock is compared to a parameter is a strong practical restriction, and is dedicated to systems that have some cyclic, repetitive behavior.\ea{OK? Délicat de critiquer trop frontalement notre précédent papier mais…}\mr{ça va c'est cool quand meme}\dl{Nice.}

\paragraph{Contribution}
In this work, we propose an original subclass of parametric timed automata, with interesting practical results.
We restrict the expressive power by disallowing guards in the model, therefore leaving the model with only invariants.

On the one hand, we show that this model of PTAs with only invariants (\PTAis{}) is at least as expressive as the original PTAs, and therefore inherits its notorious undecidability results.

On the other hand, by restraining the shape of the constraints in these invariants, giving PTAs with only invariants and upper-bound constraints (\PTAius{}), we get decidability results independently of the number of clocks or parameters used.
In addition, we show that we can synthesize the exact set of parameters for which reachability (\EF{}) 
properties hold.
This result is particularly welcome, as existing classes for which decidability of the emptiness problems hold does usually not guarantee the possibility to perform synthesis: the best-known existing subclass of PTAs, \ie{} L/U-PTAs, benefit from decidability results~\cite{HRSV02,BlT09} but synthesis cannot be achieved, even over integer-valued parameters~\cite{JLR15}.

Our formalism of \PTAius{} is the first of its kind to allow for exact synthesis over unbounded, rational-valued parameters (in contrast to~\cite{HRSV02,BlT09,ALR19}) without imposing conditions on the number of clocks or parameters (in contrast to~\cite{BO14,BBLS15}), nor imposing frequent resets (in contrast to~\cite{ALR19}).
This makes this formalism promising, together with a still interesting expressive power.
In fact, we show that for more complex properties (\eg{} nested TCTL formulas), \PTAius{} become undecidable, which shows that our formalism is far from featuring a trivial expressiveness.\mr{super}
We also exemplify our formalism on a case study, where we model a data streaming protocol using \PTAius{}.

\paragraph{Outline}
\cref{section:preliminaries} recalls the necessary preliminaries, introduces the class of PTAs without guards (\PTAis{}) and the problems of interest.
\cref{section:undecidability} proves that reachability 
is undecidable for \PTAi{}.
\cref{section:decidability} introduces an additional restriction (\PTAius{}), and proves decidability of the emptiness problems of reachability
, together with the possibility to perform synthesis.
In contrast, we show that TCTL-emptiness is undecidable for \PTAius{}, making it an expressive formalism at the border between decidability and undecidability.
\cref{section:casestudy} exemplifies our formalism on a case study.
\cref{section:conclusion} concludes the paper and proposes some perspectives.

\section{Preliminaries}\label{section:preliminaries}

\LongVersion{
\subsection{Clocks, parameters and parametric clock constraints}
}


We assume a set~$\Clock = \{ \clock_1, \dots, \clock_\ClockCard \} $ of \emph{clocks}, \ie{} real-valued variables that evolve at the same rate.
A clock valuation is\LongVersion{ a function}\dl{ça ne semble pas très bien marcher l'espace initial dans LongVersion}\ea{en fait longeversion + commentaires c'est moche, mais sans les commentaires, tout va bien. Bon, j'ai simplifié puisque de toute façon on ne fait que du long ici}
$\clockval : \Clock \rightarrow \grandrplus$.
\LongVersion{We identify a clock valuation~$\clockval$ with the point $(\clockval(\clock_1), \dots, \clockval(\clock_{\ClockCard}))$ of $\grandrplus^\ClockCard$.
}
We write $\ClocksZero$ for the clock valuation assigning $0$ to all clocks.
Given $d \in \grandrplus$, $\clockval + d$ \ShortVersion{is}\LongVersion{denotes the valuation} \st{} $(\clockval + d)(\clock) = \clockval(\clock) + d$, for all $\clock \in \Clock$.
Given $\resets \subseteq \Clock$, we define the \emph{reset} of a valuation~$\clockval$, denoted by $\reset{\clockval}{\resets}$, as follows: $\reset{\clockval}{\resets}(\clock) = 0$ if $\clock \in \resets$, and $\reset{\clockval}{\resets}(\clock)=\clockval(\clock)$ otherwise.

We assume a set~$\Param = \{ \param_1, \dots, \param_\ParamCard \} $ of \emph{parameters}\LongVersion{, \ie{} unknown constants}.
A \emph{parameter valuation} $\pval$ is\LongVersion{ a function} $\pval : \Param \to \grandqplus$.

We assume
	${\compOp} \in \{<, \leq, =, \geq, >\}$
	and
	${\compOpLeq} \in \{<, \leq\}$.
A \emph{parametric clock constraint}~$\pcc$ is a constraint over $\Clock \cup \Param$ defined by a set of inequalities of the form
	 $\clock \compOp \sum_{1 \leq i \leq \ParamCard} \alpha_i \param_i + d$, with $\alpha_i \in \{0, 1 \}$ and $d \in \grandz$.
Given~$\pcc$, we write~$\clockval\models\pval(\pcc)$ if 
the expression obtained by replacing each~$\clock$ with~$\clockval(\clock)$ and each~$\param$ with~$\pval(\param)$ in~$\pcc$ evaluates to true.
%


%

\LongVersion{
\subsection{Parametric timed automata}
}

Let $\AP$ be a set of atomic propositions.
We first recall PTAs~\cite{AHV93}.

\begin{definition}\label{def:PTA}
	A PTA
	$\A$ is a tuple \mbox{$\A = (\Actions, \Loc, \Lab, \locinit, \Clock, \Param, \invariant, \Edges)$}, where:
  \begin{itemize}
		\item $\Actions$ is a finite set of actions,
		\item $\Loc$ is a finite set of locations,
    \item $\Lab$ is a label function~$\Lab : \Loc \to 2^\AP$,
		\item $\locinit \in \Loc$ is the initial location,
		\item $\Clock$ is a finite set of clocks,
		\item $\Param$ is a finite set of parameters,
		\item $\invariant$ is the invariant, assigning to every $\loc\in \Loc$ a parametric clock constraint $\invariant(\loc)$,
		\item $\Edges$ is a finite set of edges (or transitions)  $\edge = (\loc,\guard,\action,\resets,\loc')$
		where
		$\loc,\loc'\in \Loc$ are the source and target locations, $\action \in \Actions$, $\resets\subseteq \Clock$ is a
		set of clocks to be reset, and
		the guard $\guard$ is a parametric clock constraint.
\end{itemize}
\end{definition}

Given\LongVersion{ a parameter valuation}~$\pval$, we denote by $\valuate{\A}{\pval}$ the non-parametric structure where all occurrences of a parameter~$\param_i$ have been replaced by~$\pval(\param_i)$.
We denote as a \emph{timed automaton} any structure $\valuate{\A}{\pval}$.\footnote{%
	Technically and strictly speaking, we should use a rescaling of the constants to avoid comparisons of clocks with rationals: by multiplying all constants in $\valuate{\A}{\pval}$ by the least common multiple of their denominators, we obtain an equivalent (integer-valued) TA\LongVersion{, as defined in \cite{AD94}}.}
A \emph{bounded} PTA{} is a PTA{} with a bounded parameter domain that assigns to each parameter a minimum integer bound and a maximum integer bound.
That is, each parameter~$\param_i$ ranges in an interval $[a_i, b_i]$, with $a_i,b_i \in \grandn$.
Hence, a bounded parameter domain is a hyperrectangle of dimension $\ParamCard$.

Let us first recall the concrete semantics of TAs.

\begin{definition}[Concrete semantics of a TA]
	Given a PTA $\A = (\Actions, \Loc, \Lab, \locinit, \Clock, \Param, \invariant, \Edges)$,
	and a parameter valuation~\(\pval\),
	the concrete semantics of $\valuate{\A}{\pval}$ is given by the timed transition system $(\States, \sinit, \flecheRel)$, with
	\begin{itemize}
		\item $\States = \{ (\loc, \clockval) \in \Loc \times \grandrplus^\ClockCard \mid \clockval \models \valuate{\invariant(\loc)}{\pval} \}$, 
		\item $\sinit = (\locinit, \ClocksZero) $
		\item  $\flecheRel$ consists of the discrete and (continuous) delay transition relations:
				\begin{itemize}
			\item discrete transitions: $(\loc,\clockval) \longueflecheRel{\edge} (\loc',\clockval')$, 
				if $(\loc, \clockval) , (\loc',\clockval') \in \States$, there exists $\edge = (\loc,\guard,\action,\resets,\loc') \in \Edges$, $\clockval'= \reset{\clockval}{\resets}$, and $\clockval \models \valuate{\guard}{\pval}$. 
			\item delay transitions: $(\loc,\clockval) \longueflecheRel{d} (\loc, \clockval+d)$, with $d \in \grandrplus$, if $\forall d' \in [0, d], (\loc, \clockval+d') \in \States$.
		\end{itemize}
	\end{itemize}
\end{definition}

    Moreover we write $(\loc, \clockval)\longuefleche{\edge} (\loc',\clockval')$\LongVersion{ for a combination of a delay and discrete transition where
	$((\loc, \clockval), \edge, (\loc', \clockval')) \in \flecheRel$} if
		$\exists d, \clockval'' :  (\loc,\clockval) \longueflecheRel{d} (\loc,\clockval'') \longueflecheRel{\edge} (\loc',\clockval')$.


Given a TA~$\valuate{\A}{\pval}$ with concrete semantics $(\States, \sinit, \flecheRel)$, we refer to the states of~$\States$ as the \emph{concrete states} of~$\valuate{\A}{\pval}$.
A \emph{run} of~$\valuate{\A}{\pval}$ is a possibly infinite alternating sequence of states of $\valuate{\A}{\pval}$ and edges starting from the initial state $\sinit$ of the form
$\sinit \longuefleche{\edge_0} \state_1\longuefleche{\edge_1} \cdots \longuefleche{\edge_{m-1}} \state_m \longuefleche{\edge_{m}} \cdots$, such that for all
$i = 0, 1, \dots$, $\edge_i \in \Edges$, and $(\state_i , \edge_i , \state_{i+1}) \in \flecheRel$.
Given\LongVersion{ a state}~$\state=(\loc, \clockval)$, we say that $\state$ is reachable
if $\state$ appears in a run of $\valuate{\A}{\pval}$,
or simply that~$\loc$ is reachable in~$\valuate{\A}{\pval}$, if there exists a state $(\loc,\clockval)$ that is reachable.
By extension, we say that a label $\Label$ is reachable in~$\valuate{\A}{\pval}$ if there exists a state $(\loc,\clockval)$ that is reachable such that $\Label \in \Lab(\loc)$.

Given a parameter valuation~$\pval$ and a run of~$\valuate{\A}{\pval}$ $\varrun=(\loc_{0}, \clockval_{0})\longuefleche{\edge_{0}}\cdots\longuefleche{\edge_{i-1}}(\loc_{i}, \clockval_{i})\longuefleche{\edge_{i}}(\loc, \clockval)$ we define the length of a run as the number of edges in~$\varrun$.\ea{j'ai viré la notation $|\varrun|$, non utilisée}


A \emph{maximal} run is a run that\LongVersion{ is} either \LongVersion{infinite (\ie{} }contains an infinite number of discrete transitions\LongVersion{)}, or that cannot be extended by a discrete transition.
Given a run~$\rho$ of~$\valuate{\A}{\pval}$, $\Time(\rho)$ gives the total sum of the delays~$d$ along~$\rho$.


\subsection{A new syntactic restriction}

We now introduce the first main restriction of our formalism, that consists in removing guards from PTAs.

\begin{definition}
	A \emph{PTA with only invariants (\PTAi{})} is a PTA where, in each transition, $\guard$ is always true, \ie{} is an empty set of inequalities.
\end{definition}

\subsection{Timed CTL}

TCTL~\cite{alur-ic-93} is the quantitative extension of CTL where temporal modalities are augmented with constraints on duration.
Formulae are interpreted over TTS.

Given~$\atomicprop \in \AP$ and $c\in\grandn$, a TCTL formula is given by the following\LongVersion{ grammar}:
\[
  \formule::= \  \top \ |\  \atomicprop \ |\  \neg{\formule}\ |\ \formule\wedge\formule\  |\  \styleTCTL{E}\formule\styleTCTL{U}_{\compOp c}\formule \ | \  \styleTCTL{A}\formule\styleTCTL{U}_{\compOp c}\formule
\]

$\styleTCTL{A}$ reads ``always'',
$\styleTCTL{E}$ reads ``exists'',
and
$\styleTCTL{U}$ reads ``until''.

Standard abbreviations include Boolean operators as well as~$\EF_{\compOp c}\formule$
for $\styleTCTL{E}\top \styleTCTL{U}_{\compOp c}\formule$, $\styleTCTL{AF}_{\compOp c}\formule$
for $\styleTCTL{A}\top \styleTCTL{U}_{\compOp c}\formule$ and $\EG_{\compOp c}\formule$
for $\neg\styleTCTL{AF}_{\compOp c}\neg\formule$.
($\styleTCTL{F}$ reads ``eventually'' while $\styleTCTL{G}$ reads ``globally''.)

\begin{definition}[Semantics of TCTL]
Given a TA $\pval(\A)$, the following clauses define when a state $\state_i$ of its TTS $(\States, \sinit, \flecheRel)$
satisfies a TCTL formula $\formule$, denoted by $\state_i\models\formule$,
by induction over the structure of $\formule$ (semantics of Boolean operators is omitted):
\begin{ienumeration}
 \item $\state_i\models\styleTCTL{E}\formule\styleTCTL{U}_{\compOp c}\Psi$
 if there is a maximal run~$\rho$ in~$\pval(\A)$
 with~$\sigma=\state_i\longuefleche {\edge_i} \cdots \longuefleche{\edge_{j-1}} \state_j$ ($i<j$)
 a prefix of~$\rho$
 \st{}~$\state_j\models\Psi$,
 $\Time(\sigma)\compOp c$, and if~$\forall k$ \st{} $i\leq k<j$,~$\state_k\models\formule$, and
 \item $\state_i\models\styleTCTL{A}\formule\styleTCTL{U}_{\compOp c}\Psi$
  if for each maximal run~$\rho$ in~$\pval(\A)$
  there exists~$\sigma=\state_i\longuefleche {\edge_i} \cdots \longuefleche{\edge_{j-1}} \state_j$ ($i<j$)
  a prefix of~$\rho$
 \st{} $\state_j\models\Psi$,
 $\Time(\sigma)\compOp c$, and if~$\forall k$ \st{} $i\leq k<j$,~$\state_k\models\formule$.
\end{ienumeration}
\end{definition}

In $\styleTCTL{E}\formule\styleTCTL{U}_{\compOp c}\Psi$
the classical until is extended by requiring that $\formule$ be
satisfied within a duration (from the current state) verifying the constraint ``$\compOp c$''.
Given~$\pval$, a \PTAiu{} $\A$ and a TCTL formula~$\formule$,
we write~$\pval(\A) \models \formule$ when~$\sinit \models\formule$.

We define \emph{flat TCTL} as the subset of TCTL where, in $\styleTCTL{E}\formule\styleTCTL{U}_{\compOp c}\formule$ and $\styleTCTL{A}\formule\styleTCTL{U}_{\compOp c}\formule$, $\formule$ must be a formula of propositional logic (a Boolean combination of atomic propositions).

\subsection{Problems}

In this paper, we address the following problems:

\smallskip

\defProblem
	{TCTL-emptiness}
	{a \PTAi{}~$\A$ and a TCTL formula~$\varphi$}
	{is the set of valuations~$\pval$ such that $\valuate{\A}{\pval} \models \varphi$ empty?}

\defProblem
	{TCTL-synthesis}
	{a \PTAi{}~$\A$ and a TCTL formula~$\varphi$}
	{synthesize the set of valuations~$\pval$ such that $\valuate{\A}{\pval} \models \varphi$.}


We will focus notably on the TCTL formula ``\EF{}'' expressing \emph{reachability}~\cite{AD94}.
That is, \EF{}-emptiness asks whether the set of parameter locations for which a given location is reachable for at least one run is \emph{empty} or not.
Similarly, \EF{}-synthesis asks to \emph{synthesize} these valuations.

\section{The power of invariants in PTAs}\label{section:undecidability}

In this section, we show that the expressive power of invariants in PTAs is surprisingly high:
in fact, we show that a PTA with guards but without invariants can be transformed to an equivalent \PTAi{}.
As most undecidability results for PTAs hold even without invariants, our transformation shows that \PTAi{} are (at least) as expressive as PTAs---and therefore as undecidable too.
Notably, the simplest problem for PTAs (\EF-emptiness) is undecidable for \PTAis{}.

\subsection{Transforming guards into invariants}

Let us describe our transformation from a PTA $\A$ without invariants to a \PTAi{}~$\T{\A}$.
For each edge~$\edge = (\loc_1,\guard,\action,\resets,\loc_2)$ of \A{}, we add in~$\T{\A}$ a new location~$\loc_1'$ with invariant~$\invariant(\loc_1')=\guard$ and replace~$\edge$ with a transition that is always true from~$\loc_1$ to~$\loc_1'$ with action~$\action$ and no reset: $\edge' = (\loc_1,\CTrue,\action,\emptyset,\loc_1')$. Then we add a unique transition from~$\loc_1'$ to~$\loc_2$ that is always true, without action and with the original resets~$\resets$ of~$\edge$: $\edge'' = (\loc_1',\CTrue,\epsilon,\resets,\loc_2)$ ($\epsilon$ denotes the silent action; note that actions do not matter much in our setting anyway as we are concerned with reachability properties).

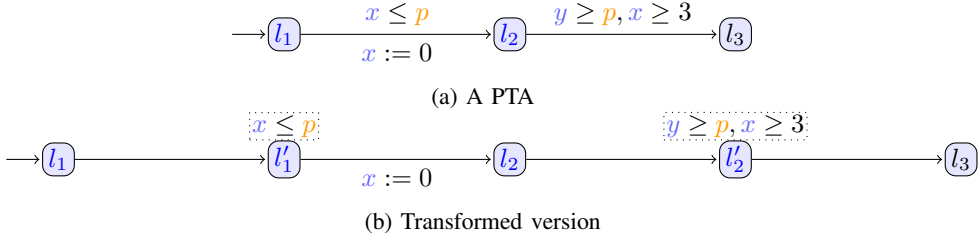
\begin{figure*}[tb]

\begin{subfigure}[b]{0.95\textwidth}
	 \centering
	\begin{tikzpicture}[shorten >=1pt, node distance=1.3cm and 3cm, on grid, auto]
	\node[location, initial, initial text={$\styleclock{x}:=0$}] (A)   {$\color{blue}l_{1}$};
	\node[location]          (B)   [right=of A]           {$\color{blue}l_{2}$};
	\node[location]          (C)   [right=of B]           {$l_{3}$};

	\path[->]
						(A)  edge   node[above] {$\styleclock{\clock}\leq\styleparam{\param}$}
						node[below] {$\styleclock{\clock}:= 0$} (B)
						(B)  edge   node[above] {$\styleclock{\clocky}\geq\styleparam{\param}, \styleclock{\clock}\geq 3$}   (C)	;

	\end{tikzpicture}
		\caption{A PTA}
		\label{fig:PTAPTAi:PTA}

\end{subfigure}

\begin{subfigure}[b]{0.95\textwidth}
	 \centering
	\begin{tikzpicture}[shorten >=1pt, node distance=1.3cm and 3cm, on grid, auto]
	\node[location, initial, initial text={$\styleclock{x}:=0$}] (A)   {$\color{blue}l_{1}$};
	\node[location] (A1) [right=of A]  {$\color{blue}l_{1}'$};
	\node [invariant,above] at (A1.north) {$\styleclock{x}\leq \styleparam{\param}$};
	\node[location]          (B)   [right=of A1]           {$\color{blue}l_{2}$};
	\node[location]          (B1)   [right=of B]           {$\color{blue}l_{2}'$};
		\node [invariant,above] at (B1.north) {$\styleclock{\clocky}\geq\styleparam{\param}, \styleclock{\clock}\geq 3$};
	\node[location]          (C)   [right=of B1]           {$l_{3}$};

	\path[->]
						(A)  edge   node[above] {} (A1)
						(A1) edge node[below] {$\styleclock{\clock}:= 0$} (B)
						(B)  edge   node[above] {}   (B1)
						(B1)  edge   node[above] {}   (C)	;

	\end{tikzpicture}
		\caption{Transformed version}
		\label{fig:PTAPTAi:transformed}
\end{subfigure}

 \caption{An example of PTA without invariant and its equivalent \PTAi{}.}
 \label{fig:PTAPTAi}
 \end{figure*}

\begin{example}
	An example of this transformation is given in \cref{fig:PTAPTAi}.
	The transition (say~$\edge$) from~$\loc_1$ to~$\loc_2$ in \cref{fig:PTAPTAi:PTA} is translated into
	\begin{inparaenum}[1)]%
		\item a new transition from~$\loc_1$ to a new location~$\loc_1'$ with as invariant the guard of the original transition~$\edge$, \ie{} $\clock \leq \param$, and
		\item a new transition from~$\loc_1'$ to~$\loc_2$ with the same reset as the one of the original transition~$\edge$, \ie{} $\clock := 0$.
	\end{inparaenum}
	This translation is exemplified in \cref{fig:PTAPTAi:transformed}.

	The guard on the transition from~$\loc_2$ to~$\loc_3$ is translated similarly.
\end{example}

\subsection{Characterization of the transformation}

We show that, for any run of~$\valuate{\A}{\pval}$, there exists in~$\valuate{\T{\A}}{\pval}$ a run twice as long, whose states of index~$2 \times i$ are identical to states of index~$i$ of the original run, for each~$i$ between~0 and the length of the run minus~1.

\begin{lemma}\label{lemma:runequivPTAsi}
Let~$\A$ be a PTA without invariant, and~$\pval$ a parameter valuation.
There is a run~$\varrun=(\loc_{0}, \clockval_{0})\longuefleche{\edge_{0}}\cdots\longuefleche{\edge_{i-1}}(\loc_{i}, \clockval_{i})\longuefleche{\edge_{i}}(\loc, \clockval)\cdots$ in~$\valuate{\A}{\pval}$ iff there is a run~$\varrun'=(\loc_{0}, \clockval_{0})\longuefleche{\edge_{0}'}(\loc_{0}', \clockval_{0}')\longuefleche{\edge_{0}''}\cdots\longuefleche{\edge_{i-1}''}(\loc_{i}, \clockval_{i})\longuefleche{\edge_{i}'}(\loc_{i}', \clockval_{i}')\longuefleche{\edge_{i}''}(\loc, \clockval)\cdots$
in~$\valuate{\T{\A}}{\pval}$.
\end{lemma}
\begin{proof}
Let~$\varrun$ be a run of~$\valuate{\A}{\pval}$ ending in a concrete state~$(\loc,\clockval)$.
We build by induction on~$n$, a run~$\varrun'$ in $\valuate{\T{\A}}{\pval}$ of length~$2n$
 taking the same sequence of edges as~$\varrun$ \wrt\ our transformation and ending in the same concrete state%
 \footnote{Note that the fact that the length is even is a consequence of the construction: with two edges, first from~$\loc$ to~$\loc''$ and the second from~$\loc''$ to~$\loc'$, if the former can be taken then~$\invariant(\loc'')$ is satisfied, and the run cannot stay forever in~$\loc''$ because of~$\invariant(\loc'')$ and is forced to take the latter to~$\loc'$.}.

If~$n = 0$, then~$\varrun'$ consists only of the initial location of~$\T{\A}$ which has no invariant, so we can stay there forever as in the initial location of~$\A$.
So any run of length~0 of $\valuate{\T{\A}}{\pval}$ is a run of $\valuate{\A}{\pval}$ and conversely.

Suppose now that we have built~$\varrun'$ for size~$n$ and consider a run~$\varrun$ with~$n+1$ edges.
Then~$\varrun$ consists of a run~$\varrun_1$, ending in~$(\loc_1, \clockval_1)$ with~$n$ edges followed by a delay~$d$
and finally a discrete transition along the edge~$\edge$ to the concrete state~$(\loc_2, \clockval_2)$.
From the induction hypothesis, we can build an equivalent run $\varrun'_1$ in~$\T{\A}$ of length~$2n$ ending in~$(\loc_1, \clockval_1)$,
Let~$\clockval'_1$ be the clock valuation obtained from~$\clockval_1$ after the delay~$d$.
By construction, if constraints defined by the guard of~$\edge$ are satisfied by~$\clockval'_1$ then
in $\varrun'_1$, we can take the transition~$\edge'$ without guards from~$\loc_1$ to~$\loc_1'$
as~$\clockval'_1\models\pval(\invariant(\loc_1'))$.
Once in~$\loc_1'$, we cannot stay forever because of~$\invariant(\loc_1')$. We can also immediately in a $0$-delay take the transition~$\edge''$ from~$\loc_1'$ to~$\loc_2$ and  clocks in $\Clock$ are reset so~$\clockval_2=\reset{\clockval'_1}{\resets}$, and we obtain a run of length~$2(n+1)$ in~$\valuate{\T{\A}}{\pval}$ ending in~$(\loc_2, \clockval_2)$.

For the other direction, starting from a run in $\T{\A}$, the initial step of the induction is similar.
Let~$\varrun'$ be a run of~$\valuate{\T{\A}}{\pval}$ of length~$2(n+1)$
ending in a concrete state~$(\loc_2,\clockval_2)$. Then~$\varrun'$ consists of a run~$\varrun'_1$,
ending in~$(\loc_1, \clockval_1)$
with~$2n$ edges followed by a first delay~$d_1$, then a discrete transition~$\edge'$ to~$\loc_1'$, and a possible delay~$d_2$ and finally a discrete transition~$\edge''$ to~$\loc_2$. Let~$\edge$ be the edge in~$\A$ corresponding to~$\edge', \edge''$ \wrt{} our construction of~$\T{\A}$, with guard~$\guard=\invariant(\loc_1')$ and the same resets as in
$\edge''$.
Suppose now that we have built by induction hypothesis~$\varrun$ in~$\valuate{\A}{\pval}$ for size~$n$ equivalent to a run~$\varrun'_1$
in~$\valuate{\T{\A}}{\pval}$ ending in~$(\loc_1, \clockval_1)$,
Let~$\clockval'_1$ be the clock valuation obtained after the delay~$d_1$ from~$\clockval_1$ and~$\clockval''_1$ after the delay~$d_2$ from~$\clockval'_1$.
By construction, if constraints defined by~$\invariant(\loc_1')$ are satisfied by~$\clockval'_1$ then
$\clockval'_1\models\pval(\guard)$. The first transition~$\edge'$ in~$\valuate{\T{\A}}{\pval}$ to~$\loc_1'$ can be taken, similarly~$\edge$ can already be taken in~$\valuate{\A}{\pval}$. After the delay~$d_2$, we still have~$\clockval''_1\models\invariant(\loc_1')$ therefore we still have~$\clockval''_1\models\pval(\guard)$.
The second transition~$\edge''$ in~$\valuate{\T{\A}}{\pval}$ to~$\loc_2$ can be taken, similarly~$\edge$ can still be taken in~$\valuate{\A}{\pval}$.
Clocks are reset along~$\edge$ so~$\clockval_2=\reset{\clockval''_1}{\resets}$ and we obtain a run of length~$n$ in~$\valuate{\A}{\pval}$ ending in~$(\loc_2, \clockval_2)$.
\end{proof}
\subsection{Undecidability for \PTAis{}}

\begin{theorem}
	\VeryLongVersion{The }\EF{}-emptiness \VeryLongVersion{problem }is undecidable for \PTAis{}.
\end{theorem}
\begin{proof}
From \cref{lemma:runequivPTAsi}, for any \VeryLongVersion{parameter }valuation~$\pval$, reachability of a location in $\valuate{\A}{\pval}$ and $\valuate{\T{\A}}{\pval}$ is equivalent.
Therefore, \EF{}-emptiness holds for $\A$ iff \EF{}-emptiness holds for $\T{\A}$.
As \EF{}-emptiness is undecidable for PTAs without invariant~\cite{AHV93}, \EF{}-emptiness is undecidable for \PTAis{}.
%
%
\end{proof}

\section{A new decidable subclass}
\label{section:decidability}

We now consider \PTAis{} \VeryLongVersion{with only invariants of the form
$\clock \compOpLeq \sum_{1 \leq i \leq \ParamCard} \alpha_i \param_i + d$ (recall that ${\compOpLeq} \in \{ \leq, < \}$) \ie{} PTAs }with only upper-bound invariants\VeryLongVersion{ (\PTAiu{})}.

\begin{definition}
	A \emph{PTA with only upper-bound invariants (\PTAiu{})} is a \PTAi{} where each inequality in an invariant is of the form $\clock \compOpLeq \sum_{1 \leq i \leq \ParamCard} \alpha_i \param_i + d$.
\end{definition}

An example of \PTAiu{} is given in \cref{fig:videostreaming}.

\PTAius{} can be seen as a subclass of L/U-PTAs, a formalism for which \EF{}-emptiness is decidable \cite{HRSV02,BlT09} while \AF{}-emptiness is undecidable \cite{JLR15}.
In addition, the synthesis of (even integer-valued) parameters for which \EF{} holds in L/U-PTAs cannot be done~\cite{JLR15}.
\PTAius{} can also be seen as a subclass of \uPTAs{}~\cite{BlT09}, \ie{} L/U-PTAs with only upper-bound parameters, a formalism for which \EF{}-emptiness is decidable \cite{HRSV02,BlT09} while \AF{}-emptiness is open, and full TCTL-emptiness is undecidable~\cite{ALR18FORMATS}; in addition, \EF{}-synthesis of integer-valued parameter can be achieved~\cite{BlT09}, but the possibility to perform or not the exact synthesis of rational-valued parameters for \EF{} remains open.

The main differences between \PTAius{} and \uPTAs{} are
\begin{ienumeration}
	\item the absence of guards in \PTAius{}, and
	\item the possibility only for \uPTAs{} to involve constraints of the form $\clock > c$ or $\clock \geq c$ in clock constraints, provided $c$ is a constant (no parameter can be used as a lower-bound constraint).
\end{ienumeration}
In this section, we will see that these differences will allow not only for positive decidability results
but will also make exact synthesis possible.

\subsection{Reachability (\EF{})}

\subsubsection{\EF{}-emptiness}

We first show that, while matching the decidability of L/U-PTAs (and \uPTAs{}) for \EF{}-emptiness, the complexity of \EF{}-emptiness for \PTAiu{} is not the same as for \uPTAs{}, which is PSPACE-complete for integer parameter valuations~\cite{BlT09}; in our case, given a \PTAiu{}~$\A$ and a special parameter valuation~$\pvalone$ that sets all parameters to~$1$, it is sufficient to test in~$\pvalone(\A)$ the reachability of a given location in a $0$-delay (a run of duration~0), which is linear in the number of locations of~$\A$.\ea{plutôt que…?}
That is, we do not perform a symbolic analysis (using the region graph~\cite{AD94} or the zone graph~\cite{BY03}) of some TA, but we directly syntactically analyze our \PTAiu{}.

Formally, let $\pvalone$ be the parameter valuation such that $\forall 1 \leq i \leq \ParamCard : \pvalone(\param_i) = 1$.
In the following lemma, we will show that there exists a valuation $\pval$ such that there exists a run in $\valuate{\A}{\pval}$ reaching a given location $\locfinal$ iff there exists a 0-delay run in~$\valuate{\A}{\pvalone}$ reaching~$\locfinal$.
By 0-delay run, we mean for which the sum of the delays along the edges is~0.
This will allow us to only test 0-delay runs in $\valuate{\A}{\pvalone}$ to decide \EF{}-emptiness.


\begin{lemma}\label{lemma:EFzeroequiv}
	Let~$\A$ be a \PTAiu{} and~$\locfinal$ a goal location.
	There exists a parameter valuation~$\pval$ and a run in $\valuate{\A}{\pval}$ reaching $\locfinal$ iff there exists a 0-delay run in~$\valuate{\A}{\pvalone}$ reaching~$\locfinal$.
\end{lemma}

\begin{proof}
\begin{itemize}
	\item [$\Longrightarrow$]
		Assume there exists a parameter valuation~$\pval$ and a run~$\varrun$ in $\valuate{\A}{\pval}$ reaching $\locfinal$.
		We first show that there exists a 0-delay run $\varrun_0$ in $\valuate{\A}{\pval}$ reaching $\locfinal$ (and, in fact, going through the same locations and edges as~$\varrun$, with only the delay being replaced with~0).
		This is immediate from the syntax of \PTAius{}: since we only allow invariants of the form $\clock \compOpLeq \sum_{1 \leq i \leq \ParamCard} \alpha_i \param_i + d$, then nothing can constrain a run to spend a certain amount of time in a location.
		Therefore, $\varrun_0$ can follow the same locations and edges as in~$\varrun$ without letting any time elapse.
		This gives that there exists a 0-delay run $\varrun_0$ in $\valuate{\A}{\pval}$ reaching $\locfinal$.

		We will now show that this run $\varrun_0$ is also a run of~$\valuate{\A}{\pvalone}$.
		This is not entirely immediate, as $\valuate{\A}{\pvalone}$ and $\valuate{\A}{\pval}$ have different invariants, coming from different parameter valuations.
		Indeed, in case of invariants of the form $\clock < \param$, a 0-delay run is blocked in this location whenever $\param = 0$ (as the constraint $\clock < 0$ is never satisfiable due to the non-negative nature of clocks).
		However, by definition, $\varrun_0$ does not pass through any location with an invariant of the form $\clock < \param$, with $\pval(\param) = 0$, since this is a valid run of~$\valuate{\A}{\pval}$.
		That is, for any location~$\loc$ along~$\varrun_0$ with an invariant containing an inequality of the form $\clock < \param$, $\pval(\param) > 0$.
		We can finally conclude by observing that, in $\valuate{\A}{\pvalone}$, no such invariant blocking a 0-delay run exists since, by definition of $\valuate{\A}{\pvalone}$, all parameters evaluate to~1.
		Therefore $\varrun_0$ is also a run reaching $\locfinal$ in~$\valuate{\A}{\pvalone}$.

	\item [$\Longleftarrow$] The opposite direction is trivial.
	It suffices to pick $\pval = \pvalone$ and, since there exists a 0-delay run in~$\valuate{\A}{\pvalone}$ reaching~$\locfinal$, then there exists a run (in 0-delay) in $\valuate{\A}{\pval}$ reaching $\locfinal$.
\end{itemize}

\end{proof}

From \cref{lemma:EFzeroequiv}, we state the following theorem.

\begin{theorem}\label{theorem:EFempt}
	\EF{}-emptiness is decidable in NLOGSPACE for \PTAiu{}.
\end{theorem}
\begin{proof}
	Let~$\A$ be a PTA and~$\locfinal$ be a target location.
	From \cref{lemma:EFzeroequiv}, there exists a parameter valuation~$\pval$ and a run in $\valuate{\A}{\pval}$ reaching $\locfinal$ iff there exists a 0-delay run in~$\valuate{\A}{\pvalone}$ reaching~$\locfinal$.
	That is, it suffices to test only the existence of at least one 0-delay run in $\valuate{\A}{\pvalone}$ to decide \EF{}-emptiness in~$\A$.

	From the nature of \PTAius{}, there exists a 0-delay run in $\valuate{\A}{\pvalone}$ iff there exists in the automaton~$\valuate{\A}{\pvalone}$ seen as a graph a syntactic path from $\locinit$ to~$\locfinal$ that features no state with an invariant involving a comparison of the form~$\clock < 0$, for some~$\clock$.
	%
	We can therefore consider~$\pvalone(\A)$ as a directed graph, in which we remove all the edges to locations where there is an invariant containing a comparison of the form $\clock < 0$ for some~$\clock$.
  In this obtained oriented graph, we perform the reachability of~$\locfinal$ from~$\locinit$ which is NLOGSPACE~\cite{Papa94}, so is \EF{}-emptiness for \PTAiu{}.
\end{proof}



\subsubsection{\EF{}-synthesis}

We will show that, in order to compute \EF{}-synthesis, it suffices to test (syntactically, without semantic analysis) each automaton obtained by replacing each parameter valuation with either 0 or~1.
This is a strong result, as \EF{}-synthesis cannot be performed for L/U-PTAs with either integer or rational valued parameters~\cite{JLR15}, and can only be performed for U-PTAs over integer-valued parameters~\cite{BlT09}. 
%
\ea{ben si, non ?}\mr{non EXPTIME pour la synthèse}%
We first define an equivalence relation for parameter valuations.

\begin{definition}\label{definition:regions-EF}
	Let~$\pval, \pval'$ be two parameter valuations.
	We say that~$\pval\sim\pval'$ if, for each parameter~$\param$,
		$\pval(\param)=0$ iff~$\pval'(\param)=0$
		(\ie{} $\pval(\param)>0$ iff~$\pval'(\param)>0$).
\end{definition}

\begin{lemma}\label{lemma:efsynth}
	Let~$\A$ be a \PTAiu{} and~$\locfinal$ a goal location.
	Let $\pval,\pval'$ be two parameter valuations such that~$\pval\sim\pval'$.

	There exists a run in $\valuate{\A}{\pval}$ reaching $\locfinal$ iff there exists a 0-delay run in~$\valuate{\A}{\pval'}$ reaching~$\locfinal$.
\end{lemma}

\begin{proof}
	The proof reuses the same technique as in \cref{lemma:EFzeroequiv}.
\begin{itemize}
	\item [$\Longrightarrow$]
		Assume there exists a parameter valuation~$\pval$ and a run~$\varrun$ in $\valuate{\A}{\pval}$ reaching $\locfinal$.
		From the reasoning used in the proof of \cref{lemma:EFzeroequiv}, there exists a 0-delay run $\varrun_0$ in $\valuate{\A}{\pval}$ reaching $\locfinal$ (and, in fact, going through the same locations and edges as~$\varrun$, with only the delay being replaced with~0).

		We will now show that this run $\varrun_0$ is also a run of~$\valuate{\A}{\pval'}$.
		Following again the reasoning used in the proof of \cref{lemma:EFzeroequiv}, by definition, $\varrun_0$ does not pass through any location with an invariant of the form $\clock < \param$, with $\pval(\param) = 0$, since this is a valid run of~$\valuate{\A}{\pval}$.
		That is, for any location~$\loc$ along~$\varrun_0$ with an invariant containing an inequality of the form $\clock < \param$, $\pval(\param) > 0$.
		We can finally conclude by observing that, in $\valuate{\A}{\pval'}$, no such invariant blocking a 0-delay run exists since, from the fact that~$\pval\sim\pval'$, $\pval(\param) > 0$ iff $\pval'(\param) > 0$ for all~$\param$.
		Therefore $\varrun_0$ is also a run reaching $\locfinal$ in~$\valuate{\A}{\pval'}$.

	\item [$\Longleftarrow$] The opposite direction is similar.
	Since there exists a 0-delay run in $\valuate{\A}{\pval'}$, then following the same reasoning as above and since~$\pval\sim\pval'$, then this same 0-delay run is also a run of~$\valuate{\A}{\pval}$.
\end{itemize}
\end{proof}

From \cref{lemma:efsynth}, it suffices to test one valuation in each of the regions defined by \cref{definition:regions-EF}.
Each region being defined by $\pval(\param) = 0$ or $\pval(\param) > 0$, for each parameter~$\param$, it suffices to test both 0 and a non-zero value, \eg{}~1.
We end up with a set~$V$ of~$2^{| \Param |}$ parameter valuations.
This gives the following theorem.


\begin{theorem}\label{theorem:EF-synthesis}
	We can compute the set \EF{}-synthesis of parameter valuations for \PTAiu{} within exponential time \wrt{} the size of the input.
\end{theorem}
\begin{proof}
From \cref{lemma:efsynth}, given a \PTAiu{} $\A$ it suffices to test the existence of at least one 0-delay run for one parameter valuation~$\pval$ in each of the regions defined by \cref{definition:regions-EF}, \ie{} from the set~$V$.
From the proof of \cref{theorem:EFempt}, this can be achieved syntactically by solving a reachability problem in the graph of~$\pval(\A)$.
If the answer to the reachability problem is positive for this parameter valuation, the whole region is added to the result.
	That is, considering two parameters $\param_1$ and $\param_2$, and the valuation such that $\pval(\param_1) = 0$ and $\pval(\param_2) = 1$, the added region is $\param_1 = 0 \land \param_2 > 0$.
However, iterate similarly for all valuations in~$V$
gives~$2^{| \Param |}$ different valuated automata and we have to test the reachability for each of them. 
Therefore, to compute \EF-synthesis, we obtain a complexity exponential in time.
\end{proof}

%

This result makes the subclass of \PTAiu{} very interesting, as a subclass of PTAs where \EF{}-synthesis can be performed. Rare subclasses such as reset-update-to-parameter PTAs~\cite{ALR19} enjoy this possibility (and only on bounded parameters), while well-known L/U-PTAs enjoy the only decidability of \EF{}-emptiness while \EF{}-synthesis has been proven intractable~\cite{JLR15}.

\subsection{Undecidability of TCTL-emptiness}

While \EF{}-emptiness is decidable for \PTAiu{}, one can wonder whether this extends to the whole TCTL-emptiness problem.
We exhibit in this section a nested TCTL formula (by opposition to flat TCTL formula, \eg{} \EF{} or \AF{}), namely \notresuperformule{}~$\atomicprop$ for some atomic property~$\atomicprop$ and prove that~\notresuperformule{}-emptiness is undecidable for (possibly bounded) \PTAiu{}.
The formula \notresuperformule{} was already used to prove the TCTL-emptiness of U-PTAs in~\cite{ALR18FORMATS}.
This implies the undecidability of the whole TCTL-emptiness problem for (possibly bounded) \PTAiu{}.

\begin{theorem}\label{theorem:UPTA}
	The $\styleTCTL{EGAF}_{=0}$-emptiness problem is undecidable for bounded \PTAiu{}.
	\end{theorem}
\begin{proof}


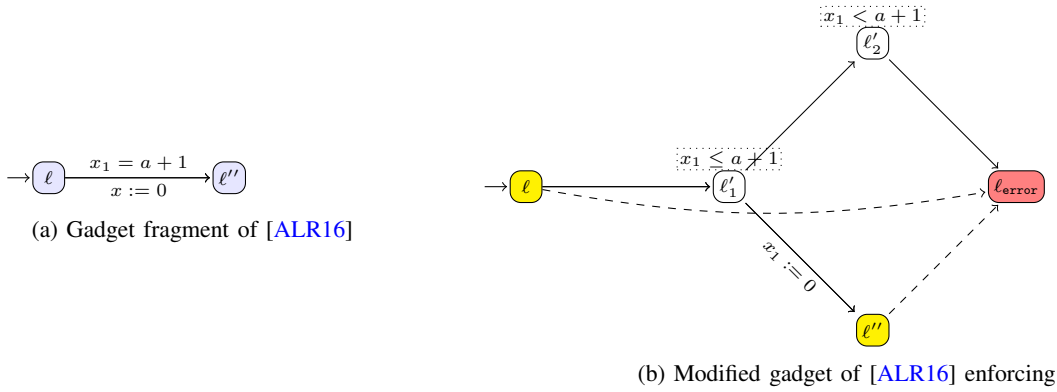
\begin{figure*}[tb]
\scriptsize
{\centering

\begin{subfigure}[c]{.3\linewidth}
	\scalebox{1}{
	\begin{tikzpicture}[shorten >=1pt, node distance=2.4cm, on grid, auto]
	\node[location, initial]   (A)      {$\loc$};
	\node[location]          (B)   [right=of A]        {$\loc''$};

	\path[->]
				(A)  edge   node  [swap, above]   {$\clock_1=\parama+1$}   (B)
				(A)  edge   node  [swap]   {$\clock:=0$}   (B);
	\end{tikzpicture}
	}
	    \caption{Gadget fragment of~\cite{ALR16ICFEM}}
	\label{figure-alr16}
    \end{subfigure}
    \hfill
   \begin{subfigure}[c]{.65\linewidth}
	\scalebox{1}{
	\begin{tikzpicture}[shorten >=1pt, node distance=2.7cm, on grid, auto]
	\node[location, gentil, initial]   (A)      {$\loc$};
	\node[location, pasgentil]          (B)   [right=of A]        {$\loc'_1$};
  \node [invariant,above] at (B.north) {$\clock_1 \leq \parama+1$};
	\node[location, gentil]          (C)   [below right=of B]        {$\loc''$};
  \node[location, pasgentil]          (D)   [above right=of B]        {$\loc'_2$};
  \node [invariant,above] at (D.north) {$\clock_1 < \parama+1$};
	\node[location, failure]          (E)   [above right=of C]        {$\locerror$};

	\path[->]
				(A)  edge   node  [swap, above]   {}   (B)
				(A)  edge   node  [swap]   {}   (B)
				(A)  edge[nongardee,bend angle=10,bend right]   node  [swap, above]   {}   (E)

				(B) edge  node[swap, above]   {}   (C)
				(B) edge  node[swap, rotate=-45, xshift=5mm]   {$\clock_1:= 0$}   (C)

				(B) edge  node[swap, above]   {}   (D)
        (D) edge node[swap, above]   {}   (E)

				(C)  edge[nongardee]   node  [swap, above]   {}   (E);
	\end{tikzpicture}
	}
	\caption{Modified gadget of~\cite{ALR16ICFEM} enforcing $\notresuperformule{}\gentil$}
	\label{figure-base-gadget}
	\end{subfigure}
	\caption{A gadget fragment and its modification into a \PTAiu{}}

}

\end{figure*}

\begin{figure*}[h!]
\scriptsize
{
	\hspace*{-2em}\begin{tikzpicture}[shorten >=1pt, node distance=2.5cm, on grid, auto]
	\node[location, initial, gentil]     (Z)   {$\loc^i$};
  \node[location, pasgentil]     (Z1)   [ right=of Z]  {$\loc_{0}^i$};
  \node [invariant,above] at (Z1.north) {$\clockz\leq 0$};

	\node[location, gentil]       (A) [right=of Z1]    {$\loc_{1}^i$};
	\node[location, pasgentil]          (B)   [above =of A]        {$\loc_{2}^i$};
  \node [invariant,above] at (B.north) {$\clock_2 \leq 1$};
  \node[location, pasgentil]          (B1)   [below right=of B]        {$\loc_{2'}^i$};
  \node [invariant,above] at (B1.north) {$\clock_2 < 1$};

  \node[location, gentil]          (C)   [above right=of B]        {$\loc_{3}^i$};
	\node[location, pasgentil]          (D)   [right=of C]        {$\loc_{4}^i$};
  \node [invariant,above] at (D.north) {$\clock_1 \leq \parama+1$};
  \node[location, pasgentil]          (D3)   [below=of D]        {$\loc_{4'}^i$};
  \node [invariant,above] at (D3.north) {$\clock_1 < \parama+1$};

  \node[location, gentil]          (D1)   [right=of D]        {$\loc_{5}^i$};
  \node[location, pasgentil, yshift=-22mm]          (D2)   [below right=of D1]        {$\loc_{6}^i$};
  \node [invariant,above] at (D2.north) {$\clockz \leq 1$};
  \node[location, pasgentil]          (D4)   [below left=of D2]        {$\loc_{6'}^i$};
  \node [invariant,above] at (D4.north) {$\clockz < 1$};

	\node[location, failure]          (E)   [below =of D3]        {$\locerror$};
	\node[location, gentil]          (I)   [right=of D2]        {$\loc^j$};

  \node[location, pasgentil]          (J)   [below =of A]        {$\loc_{7}^i$};
  \node [invariant,above] at (J.north) {$\clock_1 \leq \parama+1$};
	\node[location, pasgentil]          (G)   [right=of J]        {$\loc_{7'}^i$};
  \node [invariant,above] at (G.north) {$\clock_1 < \parama+1$};

	\node[location, gentil, yshift=-15mm]          (K)   [below right=of J]        {$\loc_{8}^i$};
	\node[location, pasgentil]          (L)   [right=of K]        {$\loc_{9}^i$};
  \node [invariant,above] at (L.north) {$\clock_2 \leq 1$};
  \node[location, pasgentil]          (L1)   [above=of L]        {$\loc_{9'}^i$};
  \node [invariant,above] at (L1.north) {$\clock_2 < 1$};
	\node[location, gentil]          (M)   [right=of L]        {$\loc_{10}^i$};

	\path[->]
				(Z)  edge   node  [swap, above]   {}   (Z1)
      	(Z1)  edge   node  [swap, above]   {}   (A)
				(Z)  edge[bend angle=18,bend right, nongardee]   node  [swap, above]   {}   (E)

				(A)  edge[nongardee]   node  [swap, above]   {}   (E)
				(A)  edge   node  [swap, above,rotate=45]   {}   (B)
				(B)  edge   node  [swap, rotate=45, xshift=-15]   {$\clock_2:=0$}   (C)
        (B)  edge   node  [swap]   {}   (B1)
        (B1)  edge   node  [swap]   {}   (E)

				(C) edge  node[swap, above]   {}   (D)
				(C)  edge[bend angle=22, nongardee]   node  [swap, above]   {}   (E)

				(D) edge  node[swap]   {$\clock_1:=0$}   (D1)
        (D) edge  node[swap]   {}   (D3)
        (D1) edge  node[swap]   {}   (D2)
        (D1)  edge[bend angle=22, nongardee]   node  [swap, above]   {}   (E)
        (D3) edge  node[swap]   {}   (E)
        (D4) edge  node[swap]   {}   (E)
        (D2) edge  node[swap]   {}   (D4)

        (J)  edge   node  [swap, above]   {}   (G)
				(G)  edge   node  [swap, above]   {}   (E)
				(D2)  edge   node  [swap]   {$\clockz:=0$}   (I)

				(A)  edge   node  [swap, above, rotate=-45, xshift=10]   {}   (J)
				(A)  edge   node  [swap, rotate=-45, xshift=15]   {}   (J)

				(J) edge  node[swap, rotate=-63]   {$\clock_1:= 0$}   (K)

				(K)  edge[nongardee]   node  [swap, above]   {}   (E)
    	(K)  edge   node  [swap, above, rotate=45]   {}   (L)
			(L)  edge   node  [swap]   {$\clock_2:=0$}   (M)
      (L)  edge   node  [swap, rotate=45, xshift=-10]   {}   (L1)
      (L1)  edge  node  [swap, above]   {}   (E)

			(M)  edge   node  [swap, rotate=45, xshift=-10]   {}   (D2)
      (M)  edge[nongardee]   node  [swap, above]   {}   (E)

				(I)  edge[bend angle=35, bend right, nongardee]   node  [swap, above]   {}   (E);

	\end{tikzpicture}

}
    \caption{increment gadget}
	\label{figure-increment-b}
\end{figure*}
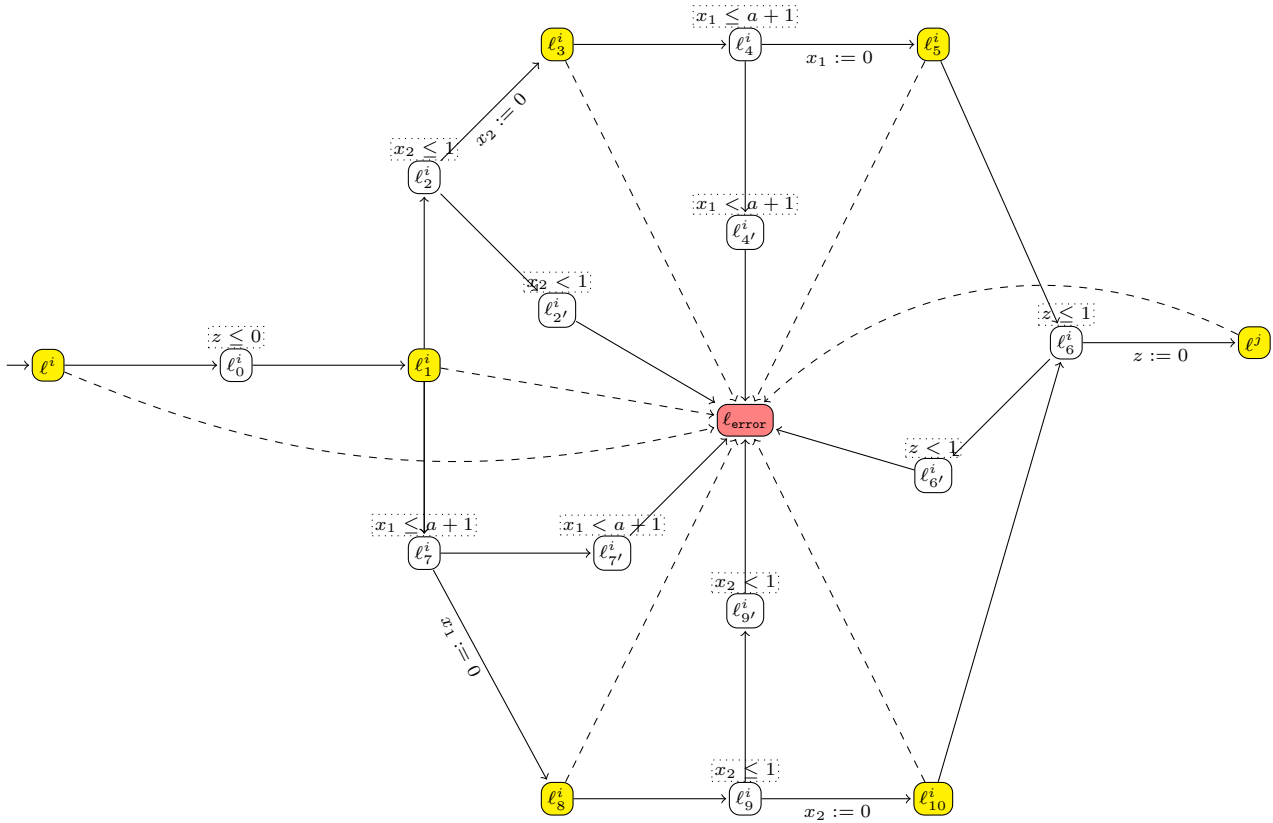


We reduce from the boundedness problem for two-counter machines (\ie{} whether the value of the counters remains bounded along the execution), which is undecidable\LongVersion{~\cite{KC10}}.
Recall that a \twoCM{} is a finite state machine with two integer-valued counters~$\cm_1, \cm_2$.
Two different instructions
are considered, we present those for~$\cm_1$, those for~$\cm_2$ are similar:
 \begin{ienumeration}
	\item when in state~$\cms_i$, increment~$\cm_1$ and go to $\cms_j$;
	\item when in state~$\cms_i$, if $\cm_1=0$ go to $\cms_k$, otherwise decrement~$\cm_1$ and go to $\cms_j$.
\end{ienumeration}%
We assume \wlogen{} that the machine halts iff it reaches a special state~$\cmshalt$.

\paragraph{General explanation of the encoding}

Let~\gentil{} and~\mechant{} be two labels.
We define a \PTAiu{} that, under some conditions, will encode the machine, and for which $\styleTCTL{EGAF}_{=0}\gentil$-emptiness holds iff the counters in the machine remain bounded.
We will reuse an encoding originally from~\cite[proof of theorem 1]{ALR16ICFEM}, and apply a few modifications.
In fact, recall that \PTAiu{} disallow the use of comparisons of the form~$\clock = \param$, or~$\clock= c$ with~$c$ a constant.

We label our transitions with: \gentil{} for the locations already present in~\cite{ALR16ICFEM} (depicted in {\gentilinline{\text{yellow}}} in our figures), and~\mechant{} for the newly introduced locations (depicted in white\LongVersion{ in our figures}).
In~\cite{ALR16ICFEM}, the gadgets use edges of the form of \cref{figure-alr16} to encode the \twoCM{} instructions.
To define a \PTAiu{}, we replace each of these edges by a special construction given in \cref{figure-base-gadget} using only inequalities of the form $\clock \leq k$ and~$\clock < k$ with~$k$ either a constant or a parameter.
Non guarded transitions are depicted as dotted edges.
We will show that a run will exactly encode the \twoCM{} if all transitions~$\clock \leq \parama+1$ (resp.\ $\clock\leq 1$) to a location labeled with~$\mechant$ are in fact taken when the clock valuation is exactly equal to~$\parama+1$ (resp.\ $1$). Those runs are further denoted by~$\rho_\gentil$.
In the transformed version given in \cref{figure-base-gadget}, due to the~$\leq$ invariant runs exist that take the guard ``too early'' (\ie{} before $\clock_1 = \parama+1$). Those are denoted by~$\rho_\mechant$.
But, in that case, observe that in $\loc_1'$, one can either take the transition to~$\loc''$ or to~$\loc_2'$ (as the invariant to satisfy is~$\clock_1<\parama+1$) and then, go to~$\locerror$.
Therefore on this gadget, \notresuperformule{}\gentil{} is true at~$\loc'$ iff the guard $\clock_1 \leq \parama+1$ from~$\loc$ to~$\loc'$ is taken at the very last
moment.
In our gadgets encoding the counters, there will be for each location with invariant~$\clock\leq k$ an associated location with invariant~$\clock < k$, with only a transition to~$\locerror$.
 Note that $\styleTCTL{AF_{=0}}\gentil{}$ is trivially true\ea{en fait seulement AF=0} in~$\loc$ and~$\loc''$ as both locations are labeled with~\gentil{}
(many runs also exist from $\loc$ to~$\locerror$ and do not encode properly the machine; they will be discarded in our reasoning later).

Our \PTAiu{}~$\A$ uses one parameter $\parama$
\LongVersion{and three parametric clocks~}\ShortVersion{,} $\clock_1, \clock_2, \clockz$.
Each state~$\cms_i$ of the \twoCM{} is encoded by a location~$\loc^i$ of~$\A$.
Each increment instruction of the \twoCM{} is encoded into a
\PTAiu{} fragment. 
The decrement instruction is a modification of the one in~\cite{ALR16ICFEM} using the same modifications as the increment gadget.

Given~$\pval$, our encoding is such that when in $\loc^i$ with $\clockval(\clockz)=0$
 then $\clockval(\clock_1)$ (resp.\ $\clockval(\clock_2)$) represents the value of the counter $\cm_1$ (resp.~$\cm_2$) encoded by~$1-\pval(\parama)\cm_1$ (resp.\ $1-\pval(\parama)\cm_2$) with~$\pval(\parama)$ small enough so~$\pval(\parama)\cm_1<1$ (resp.\ $\pval(\parama)\cm_2<1$).
The two branches in the gadgets handle both cases~$\clockval(\clock_1)>\clockval(\clock_2)$ and~$\clockval(\clock_1)\leq\clockval(\clock_2)$.

\paragraph{Increment gadget}

Depicted in \cref{figure-increment-b}. We assume~$\parama\in[0,1]$, in which case our \PTAiu{} is bounded (if $\parama$ is unbounded, then our construction proves the unbounded case).
In the following, we write~$\clockval$ as the tuple~$(\clockval(\clock_1), \clockval(\clock_2), \clockval(\clockz))$.
The initial encoding when~$\clockval(\clockz)=0$ is~$\clockval(\clock_1)=1-\pval(\parama)\cm_1, \clockval(\clock_2)=1-\pval(\parama)\cm_2, \clockval(\clockz)=0$.
  From~$\loc^i$, we prove that there is a unique run, going through the upper branch of the gadget,
  that reaches~$\loc^j$ without violating our property.
It is the one that takes each transition to a location with an invariant~$\clockz\leq 0$ at the exact moment~$\clockval(\clockz)=0$, the transition to a location with an invariant~$\clock_2\leq 1$ at the exact moment~$\clockval(\clock_2)=1$ and transition to a location with an invariant~$\clock_1\leq \parama+1$ at the exact moment~$\clockval(\clock_1)=\pval(\parama)+1$.
The other runs, that take the transitions ``too early'' are removed as they violate the property; indeed, if a run takes a transition before the ``last moment'' allowed by the invariant (\eg{} $\clock \leq 1$), then it can possibly take the successor state with invariant ($\clock < 1$) and go to \locerror{}.
That is, \notresuperformule{} does not hold, because not all runs go in 0-time to a \gentilinline{\gentil} location.

So, for each transition, many runs can take it, but we only consider from now on the only one that takes the transition at the last moment, \ie{} when the clock is exactly equal to the parameter/constant it is compared to.
The same applies at each transition.
This gives the following run for the increment gadget:

\noindent $(\gentilinline{\loc^i}, \clockval)\longuefleche{0}
(\loc_{0}^i,(1-\pval(\parama)\cm_1, 1-\pval(\parama)\cm_2, 0))\longuefleche{0}
(\gentilinline{\loc_{1}^i},(1-\pval(\parama)\cm_1, 1-\pval(\parama)\cm_2, 0))\longuefleche{\pval(\parama)\cm_2}
(\loc_{2}^i,(1-\pval(\parama)\cm_1+\pval(\parama)\cm_2, 1, \pval(\parama)\cm_2))\longuefleche{0}
(\gentilinline{\loc_{3}^i},(1-\pval(\parama)\cm_1+\pval(\parama)\cm_2, 0, \pval(\parama)\cm_2))\longuefleche{\pval(\parama)-\pval(\parama)\cm_2+\pval(\parama)\cm_1}
(\loc_{4}^i,(1+\pval(\parama), \pval(\parama)-\pval(\parama)\cm_2+\pval(\parama)\cm_1, \pval(\parama)+\pval(\parama)\cm_1))\longuefleche{0}
(\gentilinline{\loc_{5}^i},(0, \pval(\parama)-\pval(\parama)\cm_2+\pval(\parama)\cm_1, \pval(\parama)+\pval(\parama)\cm_1))\longuefleche{1-\pval(\parama)-\pval(\parama)\cm_1}
(\loc_{6}^i,(1-\pval(\parama)-\pval(\parama)\cm_1, 1-\pval(\parama)\cm_2,1))\longuefleche{0}
(\gentilinline{\loc^j}, (1-\pval(\parama)(\cm_1+1), 1-\pval(\parama)\cm_2, 0))$.

\medskip

We apply the same reasoning on the lower branch of \cref{figure-increment-b}.

\paragraph{Decrement and 0-test gadget}

The decrement and 0-test gadget, depicted in \cref{figure-decrement-b}, is similar to the one of~\cite{ALR16ICFEM} and undergoes the same modifications as in \cref{figure-increment-b}, the increment gadget.
Assume the same requirements as for the increment gadget. 
From~$\loc^i$, following the same reasoning as for the increment gadget we prove that there is a unique run, going through the upper branch of the decrement gadget,
that reaches~$\loc^j$ without violating our property.

Assume we are in\LongVersion{ a configuration}~$(\loc^i, \clockval)$
  where~$\clockval(\clockz)=0$ and suppose $\clockval(\clock_1)<1$.
  We can enter the configuration~$(\loc_{i}^1, (\clockval(\clock_1), \clockval(\clock_2), 0))$
  as the invariant~$\clockz=0$ ensures no time has elapsed; in its short form, the run that reaches~$\loc_j$ correctly, \ie{} satisfying our property \notresuperformule{} is:

\noindent $(\gentilinline{\loc^i}, \clockval)\longuefleche{0}
(\loc_{1}^i,(1-\pval(\parama)\cm_1, 1-\pval(\parama)\cm_2, 0))\longuefleche{0}
(\gentilinline{\loc_{2}^i},(1-\pval(\parama)\cm_1, 1-\pval(\parama)\cm_2, 0))\longuefleche{\pval(\parama)\cm_1}
(\loc_{3}^i,(1, 1-\pval(\parama)\cm_2+\pval(\parama)\cm_1, \pval(\parama)\cm_1))\longuefleche{0}
(\gentilinline{\loc_{4}^i},(0,1-\pval(\parama)\cm_2+\pval(\parama)\cm_1, \pval(\parama)\cm_1))\longuefleche{\pval(\parama)-\pval(\parama)\cm_1+\pval(\parama)\cm_2}
(\loc_{5}^i,(\pval(\parama)-\pval(\parama)\cm_1+\pval(\parama)\cm_2, \pval(\parama)+1, \pval(\parama)+\pval(\parama)\cm_2))\longuefleche{0}
(\gentilinline{\loc_{6}^i},(\pval(\parama)-\pval(\parama)\cm_1+\pval(\parama)\cm_2,0, \pval(\parama)+\pval(\parama)\cm_2))\longuefleche{1-\pval(\parama)\cm_2}
(\loc_{7}^i,(1-\pval(\parama)\cm_1+\pval(\parama), 1-\pval(\parama)\cm_2,\pval(\parama)+1))\longuefleche{0}
(\gentilinline{\loc^j}, (1-\pval(\parama)(\cm_1-1), 1-\pval(\parama)\cm_2, 0))$.

\medskip

We apply the same reasoning on the lower branch of \cref{figure-decrement-b}.

\paragraph{Initial gadget}

In \cref{figure-init-b}, the initial gadget ensures the same way as presented before that the counters are both initialized to~0. Recall that~$\clockval(\clock_1)=1-\pval(\parama)\cm_1$, and~$\clockval(\clock_2)=1-\pval(\parama)\cm_2$. The unique run that does not violate \notresuperformule{} reaches~$\loc_1$ exactly when~$\clockval(\clock_1)=\clockval(\clock_2)=1$, ensuring~$\cm_1=\cm_2=0$.


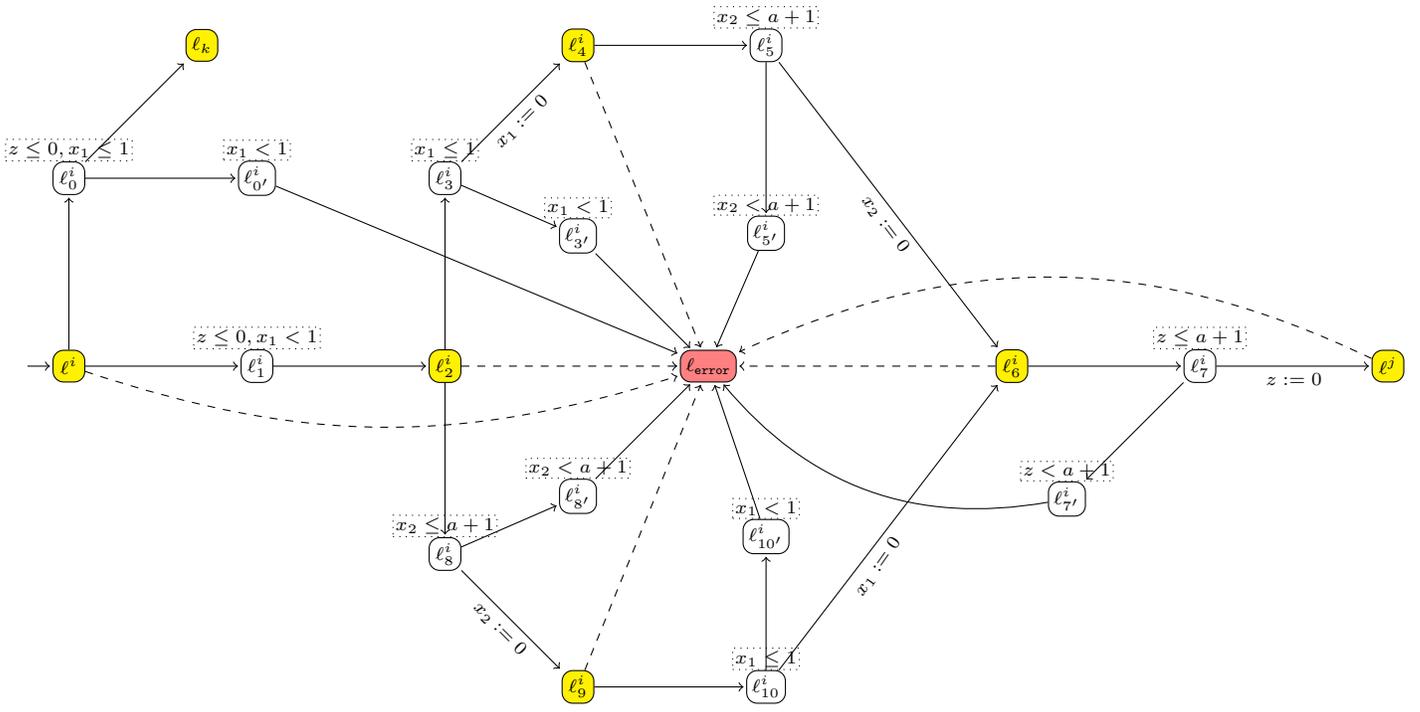
\begin{figure*}[h!]
\scriptsize
{\centering
	\hspace*{-5em}\begin{tikzpicture}[shorten >=1pt, node distance=2.5cm, on grid, auto]
	\node[location, initial, gentil]     (Z)   {$\loc^i$};
  \node[location, pasgentil]     (Z1) [above=of Z]  {$\loc_{0}^i$};
  \node [invariant,above] at (Z1.north) {$\clockz \leq 0, \clock_1\leq 1$};
  \node[location, pasgentil]     (Z2)  [right=of Z1] {$\loc_{0'}^i$};
  \node [invariant,above] at (Z2.north) {$\clock_1 < 1$};
	\node[location, gentil]       (B) [above right=of Z1]    {$\loc_{k}$};

  \node[location, pasgentil]       (A1) [right=of Z]    {$\loc_{1}^i$};
  \node [invariant,above] at (A1.north) {$\clockz \leq 0, \clock_1< 1$};
	\node[location, gentil]       (A) [right=of A1]    {$\loc_{2}^i$};

  \node[location, pasgentil]          (C1)   [above =of A]        {$\loc_{3}^i$};
  \node [invariant,above] at (C1.north) {$\clock_1 \leq 1$};
  \node[location, pasgentil]          (C2)   [below right =of C1, yshift=10mm]        {$\loc_{3'}^i$};
  \node [invariant,above] at (C2.north) {$\clock_1 < 1$};
	\node[location, gentil]          (C)   [above right=of C1]        {$\loc_{4}^i$};

	\node[location, pasgentil]          (D)   [right=of C]        {$\loc_{5}^i$};
  \node [invariant,above] at (D.north) {$\clock_2\leq \parama + 1$};
  \node[location, pasgentil]          (D1)   [below=of D]        {$\loc_{5'}^i$};
  \node [invariant,above] at (D1.north) {$\clock_2< \parama + 1$};

	\node[location, failure]          (E)   [below left=of D1, xshift=10mm]        {$\locerror$};

	\node[location, gentil]          (G)   [below right=of D1, xshift=15mm]        {$\loc_{6}^i$};
	\node[location, pasgentil]          (H)   [right=of G]        {$\loc_{7}^i$};
  \node [invariant,above] at (H.north) {$\clockz\leq \parama + 1$};
  \node[location, pasgentil]          (H1)   [below left=of H]        {$\loc_{7'}^i$};
  \node [invariant,above] at (H1.north) {$\clockz< \parama + 1$};

	\node[location, gentil]          (I)   [right=of H]        {$\loc^j$};

	\node[location, pasgentil]          (J)   [below=of A]        {$\loc_{8}^i$};
  \node [invariant,above] at (J.north) {$\clock_2\leq \parama + 1$};
  \node[location, pasgentil]          (J1)   [above right=of J, yshift=-10mm]        {$\loc_{8'}^i$};
  \node [invariant,above] at (J1.north) {$\clock_2< \parama + 1$};
	\node[location, gentil]          (K)   [below right=of J]        {$\loc_{9}^i$};
  \node[location, pasgentil]          (L)   [right=of K]        {$\loc_{10}^i$};
  \node [invariant,above] at (L.north) {$\clock_1\leq 1$};
  \node[location, pasgentil]          (L1)   [above=of L, yshift=-5mm]        {$\loc_{10'}^i$};
  \node [invariant,above] at (L1.north) {$\clock_1< 1$};

	\path[->]
				(Z)  edge   node  [swap, above]   {}   (Z1)
        (Z)  edge[bend angle=18,bend right, nongardee]   node  [swap, above]   {}   (E)

				(Z1)  edge   node  [swap, above]   {}   (Z2)
        (Z2)  edge  node  [swap, above]   {}   (E)

        (Z)  edge   node  [swap, above]   {}   (A1)
        (A1)  edge   node  [swap, above]   {}   (A)

				(Z1)  edge   node  [swap, above, rotate=45]   {}   (B)

				(A)  edge[nongardee]   node  [swap, above]   {}   (E)
				(A)  edge   node  [swap, above,rotate=45]   {}   (C1)
				(C1)  edge   node  [swap, rotate=45, xshift=-15]   {$\clock_1:=0$}   (C)
        (C1)  edge   node  [swap, rotate=45, xshift=-15]   {}   (C2)
        (C2)  edge   node  [swap, rotate=45, xshift=-15]   {}   (E)

				(C) edge  node[swap, above]   {}   (D)
				(C)  edge[bend angle=22, nongardee]   node  [swap, above]   {}   (E)

				(D) edge  node[swap, above, rotate=65]   {}   (D1)
        (D1) edge  node[swap, above, rotate=65]   {}   (E)
				(D) edge  node[swap, rotate=-53, xshift=7mm]   {$\clock_2:=0$}   (G)

				(G)  edge[nongardee]   node  [swap, above]   {}   (E)
				(G)  edge   node  [swap]   {}   (H)

				(H1)  edge[bend left, bend angle=10]   node  [swap,  above, xshift=5]   {}   (E)
				(H)  edge   node  [swap, above]   {}   (H1)
				(H)  edge   node  [swap]   {$\clockz:=0$}   (I)

				(A)  edge   node  [swap, above, rotate=-45, xshift=10]   {}   (J)

				(J) edge  node[swap, above]   {}   (J1)
        (J1) edge  node[swap, above]   {}   (E)
				(J) edge  node[swap, rotate=-45, xshift=5mm]   {$\clock_2:= 0$}   (K)

				(K)  edge[nongardee]   node  [swap, above]   {}   (E)
				(K)  edge   node  [swap, above, rotate=45, xshift=-15]   {}   (L)
        (L)  edge   node  [swap, above, rotate=45, xshift=-15]   {}   (L1)
        (L1)  edge   node  [swap, above, rotate=45, xshift=-15]   {}   (E)
				(L)  edge   node  [swap, rotate=55, xshift=-30]   {$\clock_1:=0$}   (G)

				(I)  edge[bend angle=25, bend right, nongardee]   node  [swap, above]   {}   (E);

	\end{tikzpicture}

}
	\caption{decrement gadget}
	\label{figure-decrement-b}
\end{figure*}



\begin{figure*}[h!]
\scriptsize
{\centering
	\hspace*{-5em}\begin{tikzpicture}[shorten >=1pt, node distance=2.5cm, on grid, auto]
	\node[location, initial, gentil]     (Z)   {$\loc_{\mathbf{0}}$};
  \node[location, pasgentil]     (Z1) [right=of Z]  {$\loc_{\mathbf{0}}^1$};
  \node [invariant,above] at (Z1.north) {\begin{tabular}{c}$\clockz = 0$ \\ $\clock_1\leq 1$ \\ $\clock_2\leq 1$\end{tabular}};
  \node[location, pasgentil, yshift=3mm]     (Z2)  [below right=of Z1] {$\loc_{\mathbf{0}}^2$};
  \node [invariant,above] at (Z2.north) {\begin{tabular}{c}$\clockz = 0$ \\ $\clock_1< 1$ \\ $\clock_2< 1$\end{tabular}};
  \node[location, failure]          (E)   [above right=of Z2, xshift=10mm]        {$\locerror$};

  \node[location, gentil, yshift=-10mm]       (A) [above right=of Z1]    {$\loc_{1}$};

	\path[->]
				(Z)  edge   node  [swap, above]   {}   (Z1)
        (Z)  edge[bend angle=12,bend right, nongardee]   node  [swap, above]   {}   (E)
        (Z1)  edge[bend angle=18,bend right]   node  [swap, above]   {}   (Z2)
				(Z1)  edge   node  [swap, above]   {}   (A)
        (Z2)  edge  node  [swap, above]   {}   (E)
        (A)  edge[nongardee]   node  [swap, above]   {}   (E);

	\end{tikzpicture}

}
	\caption{initialisation gadget}
	\label{figure-init-b}
\end{figure*}
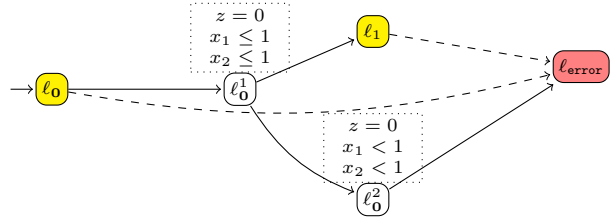


\paragraph{Simulating the 2-counter machine}

Now, let us consider the runs~$\rho_\gentil$ that take each transition to a location where there is an invariant at the very last moment; note that other runs violate the property anyway.

\begin{itemize}
	\item If the counters of the \twoCM{} remain bounded then,
	\begin{itemize}
		\item either the \twoCM{} halts by reaching~$\cmshalt$ and there exist parameter valuations~$\pval$ (typically a sufficiently small value for~$\pval(\parama)$ to encode the value of the counters during the computation).
		In the constructed \PTAiu{}, once valuated with~$\pval$ there is a (unique)\mr{oui?}\mr{ok}
		run simulating correctly the machine, reaching~\lochalt{} and staying there forever.

    In this first case, $\styleTCTL{EGAF}_{=0}\gentil$ holds for these valuations:
		hence $\styleTCTL{EGAF}_{=0}\gentil$-emptiness is false;
		\item or the \twoCM{} loops forever, never reaches~$\cmshalt$, with values of the counters remaining \emph{bounded}. There exist small parameter valuations~$\pval$ that encode the maximal value of the counters.
    In the constructed \PTAiu{}, once valuated with~$\pval$
		there is an infinite (unique)\mr{?}\mr{ok} run in the \PTAiu{} simulating correctly the machine.
		As this run is infinite, we infinitely often visit the decrement and/or the increment gadget(s).

    In this second case, $\styleTCTL{EGAF}_{=0}\gentil$ also holds for these valuations:
		hence $\styleTCTL{EGAF}_{=0}\gentil$-emptiness is again false.
      \end{itemize}
	\item Conversely, if the counters of the \twoCM{} are unbounded, then for any valuation, all runs end in~\locerror{}.
	This happens either because all the runs took on purpose an unguarded transition to~\locerror{}
	or because they blocked due to the fact that counters are unbounded, and therefore, for any arbitrarily small valuation, one of the guards will eventually block the run and send it to \locerror{} thanks to the unguarded transitions.
  That is, it is possible, \eg{} in~$\gentilinline{\loc_{5}^i}$ of \cref{figure-increment-b}, when the value of~$\clockval(\clockz)=\pval(\parama)(\cm_1+1)$ becomes strictly greater than~$1$ after a sufficient number of steps.
  It is no longer possible to take the transition to~$\loc_{6}^i$ because of the invariant~$\clockz\leq 1$
	and there is no choice other than reach~$\locerror$ again.
	Hence there is no parameter valuation
	for which $\styleTCTL{EGAF}_{=0}\gentil$ holds, so $\styleTCTL{EGAF}_{=0}\gentil$-emptiness is true.
\end{itemize}

We conclude that $\styleTCTL{EGAF}_{=0}\gentil$-emptiness is true iff the values of the counters of the \twoCM{} are unbounded.
\end{proof}

In this section, we have proved the following properties about \PTAiu{}.
Our first result here is that the \EF{}-emptiness for \PTAiu{} is \emph{less} than the same reachability problem in classical TAs without parameters.

Paradoxically, this simpler complexity for one TCTL decision problem (\EF{}) does not make \PTAiu{} a trivial subclass of (P)TAs at all.
On the contrary, we proved that the decidability of \EF{}-emptiness does not extend to the whole TCTL logic by exhibiting a TCTL formula for which deciding the \emph{emptiness} of parameter valuations satisfying it is undecidable, while model-checking TCTL logic is decidable in TAs~\cite{alur-ic-93}.

\section{Proof of concept: Case study}\label{section:casestudy}

To illustrate the usability of \PTAius{}, we describe in this section a case study modeled and verified using \VeryLongVersion{our class of }\PTAius{}.

\paragraph{Software support}
\PTAius{} are natively supported by \imitator{}~\cite{AFKS12}, which is a parametric model checker performing parameter synthesis for parametric timed automata, extended with some useful features such as synchronization, global variables, etc.
\VeryLongVersion{\imitator{} naturally supports \PTAius{} as it is a subclass of PTAs.}

\paragraph{Description}
The idea here is to model a Real-time Transport Protocol (RTP) using \PTAius{}.
RTP is a network protocol usually used to deliver video, audio over a network. RTP is mainly used in Voice over IP, teleconference and since the last few years in systems that involve media streaming.

RTP is typically running over User Datagram Protocol (UDP), which can broadcast data to several clients, and is faster as TCP (Transmission Control Protocol) as it does not provide guarantees for message delivery.

\begin{figure*}[tb]
\centering
\scalebox{.86}{
\begin{tikzpicture}[shorten >=1pt, node distance= 3cm, on grid, auto]
 \node[location] (A)   {\styleloc{idle, notSending}};
  \node [invariant,above] at (A.north) {$\styleclock{\clocky}\leq \styleparam{\param_{rced}}$};
 \node[location] (A1) [right=of A]  {$\loc_{2}$};
  \node [invariant,below] at (A1.south) {$\styleclock{x}\leq \styleparam{\param_v}$};
  \node[location, initial, initial text={$\styleclock{x}:=0$}] (A3) [above=of A1, yshift=-4em]  {$\loc_{1}$};
  \node[location]          (A2)   [right=of A1]           {$\loc_{3}$};
	  \node [invariant,below] at (A2.south) {$\styleclock{\clock}\leq\styleparam{\param_s}$};
 \node[location]          (C)   [right=of A2]           {\styleloc{idle, sending}};
   \node [invariant,above] at (C.north) {\begin{tabular}{c}$\styleclock{\clock}< \styleparam{\param_{send}}$ \\ $\styleclock{\clocky}< \styleparam{\param_{rced}}$ \end{tabular}};
   \node[location]          (D)   [below=of C,yshift=3em]           {\styleloc{askMore, sending}};
     \node [invariant,below] at (D.south) {$\styleclock{\clock}\leq \styleparam{\param_{send}}$};
   \node[location]          (E)   [below=of A,yshift=3em]           {\styleloc{askMore, notSending}};

\path[->]
				(A)  edge   node[above] {\styleact{begin}} (A1)
        (A3) edge node[right,align=center] {\styleact{start} \\ $\styleclock{\clock}:= 0, \styleclock{\clocky}:= 0$} (A1)
				(A1) edge  node[below] {\styleact{sendVideo}}  node[below] {}  (A2)
        (A2) edge node[above] {\styleact{sendSound}}  node[below] {$\styleclock{\clock}:= 0$} (C)
        (C) edge[bend left] node[above] {\styleact{interrupt}} node[below]  {$\styleclock{\clock}:=0$} (A)
        (C) edge[bend left] node[right,align=center] {\styleact{outOfData} \\ $\styleclock{\clocky}:= 0$} (D)
        (D) edge[bend left] node[right] {}(C)
        (A) edge node[below] {} (E)
        (D) edge[bend angle=5, bend left] node[below] {}(E)
        ;

\end{tikzpicture}
}
\caption{Model of a media streaming protocol}
\label{fig:videostreaming}
\end{figure*}
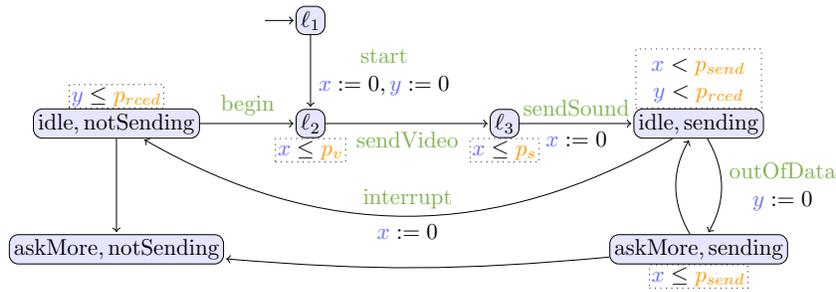

\cref{fig:videostreaming} represents a simplified version of an RTP protocol combined with a Real-Time Control Protocol (RTCP).
A server sends audio and video data to a client, and the client has the possibility to pause the data stream or ask for more data when its buffer is empty.
We use two clocks to model the protocol.
$\styleclock{\clock}$ represents the server, while~$\styleclock{\clocky}$ represents the client. In each location, the first word represents the state of the client, while the second represents the state of the server.
The automaton starts in location~${\color{blue}\styleloc{\loc_1}}$ as the client is waiting for its data stream.
On the \styleact{begin} action, the server first opens the channel for the video within~$\styleparam{\param_v}$ units of time, and the channel for the audio within~$\styleparam{\param_s}-\styleparam{\param_v}$ units of time, assuming otherwise audio and video would not be synchronized at reception by the client.
Then data is streamed for at most~$\styleparam{\param_{send}}$ units of time to prevent overflowing the bandwidth, in location~\styleloc{idle, sending}. At this moment, the server stops sending for an undetermined amount of time. In the meantime, the client's buffer is being emptied. When running \styleact{outOfData}, the client switches to location~\styleloc{askMore, sending} as the server is still sending data. \styleclock{\clocky} is reset and the system has the possibility to switch to location~\styleloc{idle, sending} again if the server is still streaming data, \ie{} the constraint~$\styleclock{\clock}<\styleparam{\param_{send}}$ is still satisfied.
While in~\styleloc{idle, sending}, the client can choose to~\styleact{interrupt} the data stream. When in location~\styleloc{idle, notSending}, the client still uses the data of the buffer, but has to request more data at some point, \ie{} while~$\styleclock{\clocky}<\styleparam{\param_{rced}}$ is satisfied. The procedure from~\styleact{start} is similar to the previously described one.

From locations~\styleloc{askMore, sending} and~\styleloc{idle, notSending} the location~\styleloc{askMore, notSending} is reachable, when the server is not streaming and the client's buffer is empty. This is the bug state of the system.
We are interested in computing the concrete parameter valuations of~$\styleparam{\param_{send}}, \styleparam{\param_{rced}}, \styleparam{\param_{s}}, \styleparam{\param_{v}}$ \st{} the system can reach the ``bad'' state~\styleloc{askMore, notSending}---that is, we aim at performing \EF{}(\styleloc{askMore, notSending)}-synthesis.

\paragraph{Experiments}
We modeled the case study in \cref{fig:videostreaming} in the input language of \imitator{}.
Experiments were conducted with \imitator{} 2.11 ``Butter Kouign-amann'', on a 2.4\,GHz Intel Core i5 processor with\,2 GiB of RAM in a VirtualBox environment running Ubuntu.\ea{running which os?}\footnote{%
  Models and results are available at \url{https://www.imitator.fr/static/ICECCS19/}
}
The synthesis time is less than 1~second with four parameters.

Applying \imitator{} to \cref{fig:videostreaming}, we obtain the following result for \EF{}(\styleloc{askMore, notSending)}-synthesis:

\vspace{-1em}$$
 	\styleparam{\param_s} \geq 0
 	\land \styleparam{\param_v} \geq 0
 	\land \styleparam{\param_{send}} > 0
 	\land \styleparam{\param_{rced}} > 0.
	$$\vspace{-1.6em}

That is, for almost all parameter valuations, there exists an execution of the system such that it reaches the bad location~\styleloc{askMore, notSending}.
This is not surprising, as it depends on the rate of data exchanged and of the connection quality to the network.\ea{ça te va ? je suis pas trop sûr…}\mr{oui}
In other words, this bug state can be reached in any case as the data stream can be blocked at any time, \ie{} the client may have to wait for the video to load.

A more interesting question is to study whether all runs of some valuations may eventually reach the bug location.
This would be worrying, as it would denote that the protocol has no chances of success for these valuations.
Therefore, we focus on \EF{}(\styleloc{askMore, notSending})-synthesis.
This time, we obtain that the set of valuations for which all runs eventually reach \styleloc{askMore, notSending} is empty, and therefore no valuation makes the protocol entirely unsuccessful.

\ea{dommage que la synthèse ne montre pas d'accumulation de contraintes (du genre $\param > \param'$, car ça aurait pu arriver, en tout cas pour \AF{}}

\section{Conclusion}\label{section:conclusion}

We proposed a new parametric timed formalism to reason about timed systems with some uncertain or unknown timing constants, with two interesting positive results.
First, the emptiness of the valuation set for which at least one run reaches a location
\ie{} \EF{}-emptiness
, is decidable in linear time 
which is better than solving the reachability problem for TAs, as it is PSPACE-complete.
Second, we showed that exact synthesis can be achieved in exponential time.

In contrast, we showed that (nested) TCTL-emptiness is undecidable, making \PTAius{}, as model-checking TCTL is decidable for TAs, a formalism at the border between decidability and undecidability.

Our formalism seems to allow for promising practical applications as shown by \cref{section:casestudy}, where we successfully modeled a simple data streaming protocol.
\VeryLongVersion{This case study is interesting as our \PTAiu{} is naturally designed to model such systems where only upper-bound timing constraints are required, as in a data streaming protocol where the client needs a continuous data flow and should not have is buffer empty, causing an interruption when reading the data stream.
The subsequent analysis allowed us to compute the timing constraints that make such a bug state reachable in the model, and further predict a scenario where these constraints are incorporated.}



\paragraph*{Future work}
%
On the theoretical side, the emptiness of some flat TCTL formulas remains open for \PTAius{}, notably \AF{}, \EG{} and \AG{}-emptiness.
Improving the complexity of \EF{}-synthesis is also an interesting direction.

More practically, we are interested in proposing dedicated efficient synthesis algorithms for \PTAius{} (independently of the underlying decidability).



\ifdefined\AuthorVersion
	\newcommand{\CCIS}{Communications in Computer and Information Science}
	\newcommand{\ENTCS}{Electronic Notes in Theoretical Computer Science}
	\newcommand{\FI}{Fundamenta Informormaticae}
	\newcommand{\FMSD}{Formal Methods in System Design}
	\newcommand{\IJFCS}{International Journal of Foundations of Computer Science}
	\newcommand{\IJSSE}{International Journal of Secure Software Engineering}
	\newcommand{\IPL}{Information Processing Letters}
	\newcommand{\JLAP}{Journal of Logic and Algebraic Programming}
	\newcommand{\JLC}{Journal of Logic and Computation}
	\newcommand{\LMCS}{Logical Methods in Computer Science}
	\newcommand{\LNCS}{Lecture Notes in Computer Science}
	\newcommand{\RESS}{Reliability Engineering \& System Safety}
	\newcommand{\STTT}{International Journal on Software Tools for Technology Transfer}
	\newcommand{\TCS}{Theoretical Computer Science}
	\newcommand{\ToPNoC}{Transactions on Petri Nets and Other Models of Concurrency}
	\newcommand{\TSE}{IEEE Transactions on Software Engineering}
	\renewcommand*{\bibfont}{\small}
	\printbibliography[title={References}]
\else
	\bibliographystyle{IEEEtran} 
	\newcommand{\CCIS}{CCIS}
	\newcommand{\ENTCS}{ENTCS}
	\newcommand{\FI}{FI}
	\newcommand{\FMSD}{FMSD}
	\newcommand{\IJFCS}{IJFCS}
	\newcommand{\IJSSE}{IJSSE}
	\newcommand{\IPL}{IPL}
	\newcommand{\JLAP}{JLAP}
	\newcommand{\JLC}{JLC}
	\newcommand{\LMCS}{LMCS}
	\newcommand{\LNCS}{LNCS}
	\newcommand{\RESS}{RESS}
	\newcommand{\STTT}{STTT}
	\newcommand{\TCS}{TCS}
	\newcommand{\ToPNoC}{ToPNoC}
	\newcommand{\TSE}{TSE}
	\bibliography{PTASI}
\fi

\end{document}